\documentclass[11pt]{article}

\RequirePackage{amssymb} 
\RequirePackage{amsmath} 
\RequirePackage{latexsym}
\RequirePackage{mathrsfs}

\DeclareMathOperator*{\supp}{supp}
\DeclareMathOperator*{\cov}{cov}

\newcommand{\R}{\ensuremath{\mathbb{R}}}
\newcommand{\Exp}{\ensuremath{\mathbb{E}}} 
\newcommand{\Prob}{\ensuremath{\mathbb{P}}} 

\usepackage{bbm}
\newcommand{\indicator}{\ensuremath{\mathbbm{1}}}

\newcommand{\normal}{\ensuremath{\mathcal{N}}}

\newcommand\independent{\protect\mathpalette{\protect\independenT}{\perp}}
\def\independenT#1#2{\mathrel{\rlap{$#1#2$}\mkern2mu{#1#2}}}

\usepackage{color}

\usepackage{wrapfig}
\usepackage{pgfplots}

\RequirePackage{amsthm}

\theoremstyle{definition}

\newtheorem{proposition}{Proposition}
\newtheorem{lemma}{Lemma}
\newtheorem{definition}{Definition}
\newtheorem{theorem}{Theorem}
\newtheorem{corollary}{Corollary}

\newcounter{partialIndepSection}
\newcommand{\aAssump}{A\arabic{partialIndepSection}}

\newtheoremstyle{theoremSuppressedNumber}{}{}{}{}{\bfseries}{.}{ }{\thmname{#1}\thmnote{ (\mdseries #3)}}
\theoremstyle{theoremSuppressedNumber}
\newtheorem{partialIndepAssump}{Assumption \aAssump \addtocounter{partialIndepSection}{1}}

\usepackage{setspace}
\usepackage[longnamesfirst]{natbib}
\usepackage{ulem}
\usepackage{subfig}
\usepackage{wrapfig}
\usepackage{float}

\pdfoutput=1
\usepackage{hyperref}
\hypersetup{
    colorlinks=true,
    linkcolor=black,
    citecolor=black,
    filecolor=black,
    urlcolor=black,
}

\title{\textbf{Interpreting Quantile Independence}\footnote{This paper is based on a portion of our previous working paper, \cite{MastenPoirier2016}. This paper was presented at the University of Wisconsin--Madison, Northwestern University, the 2017 Triangle Econometrics Conference, the 2018 North American Winter Meetings of the Econometric Society, and the 2018 Vanderbilt Conference on Identification in Econometrics. We thank audiences at those seminars and conferences, as well as Federico Bugni, Ivan Canay, Andrew Chesher, Joachim Freyberger, Joel Horowitz, Shakeeb Khan, Roger Koenker, Chuck Manski, Jim Powell, and Adam Rosen, for helpful conversations and comments.}}

\author{Matthew A. Masten\footnote{Department of Economics, Duke University,
        \texttt{matt.masten@duke.edu}} \qquad Alexandre Poirier\thanks{
    Department of Economics, University of Iowa,
    \texttt{alexandre-poirier@uiowa.edu}}
}

\begin{document}
\maketitle
\begin{abstract}
How should one assess the credibility of assumptions weaker than statistical independence, like quantile independence? In the context of identifying causal effects of a treatment variable, we argue that such deviations should be chosen based on the form of selection on unobservables they allow. For quantile independence, we characterize this form of treatment selection. Specifically, we show that quantile independence is equivalent to a constraint on the average value of either a latent propensity score (for a binary treatment) or the cdf of treatment given the unobservables (for a continuous treatment). In both cases, this average value constraint requires a kind of non-monotonic treatment selection. Using these results, we show that several common treatment selection models are incompatible with quantile independence. We introduce a class of assumptions which weakens quantile independence by removing the average value constraint, and therefore allows for monotonic treatment selection. In a potential outcomes model with a binary treatment, we derive identified sets for the ATT and QTT under both classes of assumptions. In a numerical example we show that the average value constraint inherent in quantile independence has substantial identifying power. Our results suggest that researchers should carefully consider the credibility of this non-monotonicity property when using quantile independence to weaken full independence.
\end{abstract}



\onehalfspacing
\normalsize

\newpage
\section{Introduction}\label{sec:intro}

A large literature has studied identification and estimation of structural econometric models under quantile independence rather than full independence.\footnote{Important early work includes \cite{KoenkerBassett1978}, who introduced quantile regression, \cite{Manski1975}, who studied discrete choice models under median independence, and \cite{Manski1985}, who extended that analysis to quantile independence.} This literature has faced a longstanding question: How can one substantively interpret and judge the credibility of a given set of quantile independence conditions? For example, \citet[page 733]{Manski1988} notes that ``Quantile independence restrictions sometimes make researchers uncomfortable. The assertion that a given quantile of $u$ does not vary with $x$ may lead one to ask: Why this quantile but not others?''

In this paper, we answer this question by providing a treatment assignment characterization of quantile independence conditions. Specifically, we consider the relationship between a continuous unobserved variable $U$ and an observed treatment variable $X$, which may be binary, discrete, or continuous. For example, $U$ could be an unobserved structural variable like ability, or an unobserved potential outcome. In section \ref{sec:characterizationSection} we consider the binary $X$ case. There the dependence structure between $X$ and $U$ is fully characterized by the conditional probability
\[
	p(u) = \Prob(X = 1 \mid U=u).
\]
If $U$ were observed, this would be an ordinary propensity score. Since $U$ is not observed, however, we call this a \emph{latent} propensity score; for brevity we will often refer to it simply as a propensity score. Constant propensity scores correspond to full statistical independence while non-constant propensity scores represent deviations from full independence. In this sense, any exogeneity assumption weaker than full independence allows for certain kinds of selection on unobservables. Our main theorem characterizes the set of propensity scores consistent with a set of quantile independence conditions. That is, we fully describe the kinds of selection on unobservables that quantile independence does and does not allow. This result shows that quantile independence imposes a set of constraints on \emph{average values} of the propensity score $p$. We then describe several properties of propensity scores which satisfy these constraints. Most notably, non-constant propensity scores which are consistent with a single quantile independence condition must be non-monotonic. Furthermore, if multiple isolated quantile independence conditions hold, then non-constant propensity scores must also oscillate up and down. These results do not depend on a specific econometric model, and hence apply any time one makes a quantile independence assumption. We use our main result to compare the constraints quantile independence imposes on selection with the constraints imposed by mean independence. In particular, we show that mean independence also requires non-constant propensity scores to be non-monotonic.

In section \ref{sec:ctsX} we generalize our main theorem on characterization of quantile independence to discrete and continuous treatments. As in the binary treatment case, quantile independence is equivalent to a set of average value constraints on the distribution of treatment given the unobservable. In both the discrete and continuous case, this constraint implies a non-monontonicity result. Specifically, treatment cannot be \emph{regression dependent} on the unobservable; equivalently, the distribution of treatment given unobservables cannot be stochastically monotonic in the unobservable.

To understand the restrictiveness of these non-monotonicity constraints, and therefore the plausibility of quantile independence assumptions, in section \ref{sec:LatentSelectionModels} we study two simple economic models of treatment selection: One for continuous treatments and one for binary treatments. In both models we show that standard assumptions on the economic primitives imply treatment selection rules that are monotonic in the unobservable. Therefore, by our characterization results, these selection models with the standard assumptions are incompatible with quantile independence. We then discuss how quantile independence can be retained by considering alternative assumptions on the economic primitives. Consequently, researchers using quantile independence constraints must argue that such alternative assumptions are plausible in the empirical setting under consideration.

While the primary contribution of this paper is about the interpretation of quantile independence, in section \ref{sec:treatmentEffects} we return to the binary treatment case to study the implications of our results for identification. We consider the case where an interval of quantile independence conditions hold, like $[0.25,0.75]$. We show that the latent propensity score must be flat on that interval. Our characterization result shows that the latent propensity score must also satisfy an average value constraint outside of that interval, which implies that the latent propensity score is non-monotonic outside of that interval. To isolate the identifying power of this average value constraint, we define a concept weaker than quantile independence, which we call \emph{$\mathcal{U}$-independence}. This concept simply specifies that, on the set $\mathcal{U}$, the latent propensity score equals the overall unconditional probability of being treated. Hence the latent propensity score is flat on $\mathcal{U}$. While quantile independence on the set $\mathcal{U}$ also imposes this flatness constraint, $\mathcal{U}$-independence does not constrain the average value of the latent propensity score outside of this set. Thus the difference between identified sets derived under the two assumptions can be ascribed solely to this average value constraint, which is what requires latent propensity scores to be non-monotonic.

To understand the identifying power of the average value constraint, one must first specify an econometric model and a parameter of interest. While our quantile independence results are relevant for many different models (such as those cited in our literature review below), we focus on a simple but important model: the standard potential outcomes model of treatment effects with a binary treatment. We adapt the analysis of \cite{MastenPoirier2017} to derive identified sets for the average treatment effect for the treated (ATT) and the quantile treatment effect for the treated (QTT) under either an interval of quantile independence assumptions, or under $\mathcal{U}$-independence. We then compare these identified sets in a numerical illustration. In this illustration, the identified sets are significantly larger under $\mathcal{U}$-independence, implying that the average value constraint has substantial identifying power.

\subsection*{Related Literature}\label{sec:litReview}

Quantile independence of $U$ from $X$ constraints the distribution of the unobservable $U$ conditional on the observable $X$. Our analysis is essentially a study of what these constraints on $U \mid X$ imply about the distribution of $X \mid U$. Hence it relies on the fact that, like mean independence, quantile independence treats these two variables asymmetrically. While this asymmetry has long been noted in the literature (for example, page 85 of \citealt{Manski1988book}), we are unaware of any prior study of its implications.\footnote{A large literature in statistics studies dependence concepts; for example, see \cite{Joe1997}. Many of these concepts are asymmetric. This literature studies the properties of and relations between these concepts. Quantile independence is not one of the commonly studied concepts in this literature.}  Instead, most prior research states various constraints on the joint distribution of $U$ and $X$ as a menu of options, with little guidance for choosing between them (for example, see \citealt{Manski1988} or section 2 of \citealt{Powell1994}). Mean independence is sometimes argued to be undesirable since it is not invariant to strictly monotone transformations of the variables (for example, see page 8 of \citealt{Imbens2004}). This non-invariance concern does not apply to quantile independence restrictions, or the stronger assumption of full independence.

A key point of our paper is that the choice of \emph{any} assumption weaker than full independence depends on the form of selection on unobservables one wishes to allow. We have emphasized this by characterizing the form of selection on unobservables allowed by quantile independence conditions. In principle, a similar analysis can be done for other kinds of exogeneity assumptions, like zero correlation, mean independence, or conditional symmetry.

In this paper we emphasize the interpretation of a set of quantile independence conditions which are \emph{strictly} weaker than full independence. As the most common case, a large literature has studied the identifying power of a single quantile independence condition. For example, \citet[page 732]{Manski1988} performs an identification analysis and notes that ``If, in fact, other quantiles are also independent of $x$, this information is ignored.'' To give an idea of the breadth of models where quantile independence is used, we review a subset of papers which perform similar identification analyses; we omit papers which use quantile independence conditions only as a characterization of a full independence assumption.\footnote{In particular, our results are not relevant if one is interested solely in descriptive quantile regressions, rather than causal effects or structural functions. \cite{Sasaki2015} analyzes the relationship between quantile regressions and structural functions in detail.} We also omit papers primarily on estimation theory. See \cite{QRhandbook2017} for a comprehensive overview of quantile methods.

Settings where quantile independence is used as an identifying assumption include:
binary response models with interval measured regressors (\citealt{ManskiTamer2002}),
discrete response models with exogenous regressors (\citealt{Manski1985}, \citealt{Torgovitsky2015}),
discrete response models with endogenous regressors (\citealt{BlundellPowell2004}, \citealt{Chesher2010}),
discrete games (\citealt{Tang2010}, \citealt{WanXu2014}, \citealt{Kline2015}),
IV quantile treatment effect models (\citealt{ChernozhukovHansen2005}, \citealt{Chesher2007worldCongress}, \citealt{ChernozhukovHansenWuthrich2017}),
triangular nonseparable models (\citealt{Chesher2003,Chesher2005,Chesher2007,Chesher2007worldCongress}),
generalized instrumental variable models (\citealt{ChesherRosen2017}),
panel data models (\citealt{Koenker2004}, \citealt{GalvaoKato2017}),
censored regression models (\citealt{Powell1984,Powell1986}, \citealt{HonoreKhanPowell2002}, \citealt{HongTamer2003}),
social interaction models (\citealt{BrockDurlauf2007}),
bargaining models (\citealt{MerloTang2012}),
and transformation models (\citealt{Khan2001}).

One of our main results is that quantile independence assumptions impose a \emph{non}-monotonicity condition on treatment selection. In contrast, monotonicity conditions of various kinds are often viewed as plausible, and are widely used throughout the econometrics literature. These assumptions include
monotonicity of treatment on potential outcomes (\citealt{Matzkin1994}, \citealt{Manski1997}, \citealt{AltonjiMatzkin2005}), 
monotonicity of unobservables on potential outcomes (\citealt{Matzkin2003}, \citealt{Chesher2003,Chesher2005}), 
monotonicity of an instrument on potential treatment (\citealt{ImbensAngrist1994}),
monotonicity of unobservables on potential treatment (\citealt{ImbensNewey2009}), 
mean potential outcomes conditional on an instrument are monotonic in the instrument (\citealt{ManskiPepper2000,ManskiPepper2009}), 
quantiles of potential outcomes conditional on an instrument are monotonic in the instrument (\citealt{Giustinelli2011}), 
monotonicity of potential outcomes in actions of other players in a game (\citealt{KlineTamer2012}, \citealt{Lazzati2015}), 
and many others. See \cite{ChetverikovSantosShaikh2017} for a recent survey of the econometrics literature on shape restrictions, including monotonicity. When applying the main result in our paper to the treatment effects model, our non-monotonicity result concerns the relationship between a realized treatment and an unobservable, like a potential outcome.

\section{Characterizing Quantile Independence}\label{sec:characterizationSection}

In this section, we present our main characterization result. Let $X$ be an observable random variable and $U$ an unobservable random variable. For example, in section \ref{sec:treatmentEffects} we study a treatment effects model where $X$ is a binary treatment and $U$ is an unobserved potential outcome. In this section, however, we generally remain agnostic as to the interpretation of these variables. Finally, note that all results in this section continue to hold if one conditions on an additional vector of observed covariates $W$, as is typically the case in empirical applications.

\subsection{A Class of Quantile Independence Assumptions}

Quantile independence of $U$ from $X$ is based on the well known result that statistical independence between $U$ and $X$ is equivalent to
\[
	Q_{U \mid X}(\tau \mid x) = Q_U(\tau)
\]
for all $\tau \in (0,1)$ and all $x \in \supp(X)$. Existing research typically focuses on two extreme assumptions: a single quantile independence condition holds, such as $Q_{U \mid X}(0.5 \mid x) = Q_U(0.5)$ for all $x \in \supp(X)$, or all quantile independence conditions hold (statistical independence). We study a class of assumptions which includes both of these cases. It is often more natural to work with cdfs. Say $U$ is $\tau$-cdf independent of $X$ if 
\begin{equation}\label{cdfIndependence}
	F_{U \mid X}(u \mid x) = F_U(u)
\end{equation}
holds for $u=\tau$ and for all $x \in \supp(X)$. This motivates the following definition.\footnote{See \cite{BelloniChenChernozhukov2017} and \cite{ZhuZhangXu2017} for similar generalizations of quantile independence.}

\begin{definition}
Let $\mathcal{T}$ be a subset of $\R$. Say $U$ is \emph{$\mathcal{T}$-independent} of $X$ if for all $\tau \in \mathcal{T}$, the cdf independence condition \eqref{cdfIndependence} holds for all $x \in \supp(X)$.
\end{definition}

We assume that $U$ is continuously distributed throughout this paper. In this case, we can without loss of generality normalize its distribution to be uniform on $[0,1]$. This follows since $U$ is $\tau$-cdf independent of $X$ if and only if $F_U(U)$ is $F_U(\tau)$-cdf independent of $X$, for continuously distributed $U$. For the normalized variable $F_U(U)$, equation $\eqref{cdfIndependence}$ is nontrivial only for $\tau \in (0,1)$, and hence it suffices to let $\mathcal{T}$ be a subset of $(0,1)$.

\subsection{The Characterization Result}

We now focus on the binary treatment case, $X \in \{ 0,1 \}$. We extend our results to the discrete and continuous $X$ case in section \ref{sec:ctsX}. Recall that the dependence structure between $X$ and $U$ is fully characterized by the latent propensity score
\[
	p(u) = \Prob(X=1 \mid U=u).
\]
Full statistical independence, $X \independent U$, is equivalent to this propensity score being constant:
\[
	p(u) = \Prob(X=1)
\]
for almost all $u \in \supp(U)$. Consequently, any non-constant propensity score represents a deviation from full independence. $\mathcal{T}$-independence restricts the form of these deviations. The following theorem characterizes the set of propensity scores consistent with $\mathcal{T}$-independence.

\begin{theorem}[Average value characterization]\label{thm:AvgValueCharacterization}
Suppose $U$ is continuously distributed; normalize $U \sim \text{Unif}[0,1]$. Suppose $X$ is binary with $\Prob(X=1) \in (0,1)$. Then $U$ is $\mathcal{T}$-independent of $X$ if and only if
\begin{equation}\label{eq:averageValueCondition_main}
	\frac{1}{t_2 - t_1} \int_{[t_1,t_2]} p(u) \; du = \Prob(X=1)
\end{equation}
for all $t_1, t_2 \in\mathcal{T} \cup \{ 0,1 \}$ with $t_1 < t_2$.
\end{theorem}

The proof, along with all others, is in appendix \ref{sec:proofs}. Theorem \ref{thm:AvgValueCharacterization} says that $\mathcal{T}$-independence holds if and only if for every interval with endpoints in $\mathcal{T} \cup \{ 0,1 \}$ the average value of the propensity score over that interval equals the overall average of the propensity score, since
\[
	\int_0^1 p(u) \; du = \Prob(X=1).
\]
This overall average is just the unconditional probability of being treated.

Given our assumption that $U$ is continuously distributed, specifying $U \sim \text{Unif}[0,1]$ is a normalization which simply rescales the latent propensity score's domain. If $V$ is our original continuously distributed variable and $U \equiv F_V(V)$ is our scaled variable, then the constraint \eqref{eq:averageValueCondition_main} on $p(u)$ can be translated into a constraint on the original latent propensity score via the equation $\Prob(X=1 \mid V=v) = p(F_V(v))$ for almost all $v \in \supp(V)$.

To illustrate theorem \ref{thm:AvgValueCharacterization}, suppose $\mathcal{T} = \{ 0.5 \}$ and $\Prob(X=1) = 0.5$. Here we have just a single nontrivial cdf independence condition, median independence. Figure \ref{propensityScoresFig1} plots three different propensity scores which are consistent with $\mathcal{T}$-independence under this choice of $\mathcal{T}$; that is, which are consistent with median independence. This figure illustrates several features of such propensity scores: The value of $p(u)$ may vary over the entire range $[0,1]$. $p$ does not need to be symmetric about $u = 0.5$, nor does it need to be continuous. Finally, as suggested by the pictures, $p$ must actually be nonmonotonic; we show this in corollary \ref{corr:monotonicPropensityScores} next.

\begin{figure}[t]
\centering
\includegraphics[width=52mm]{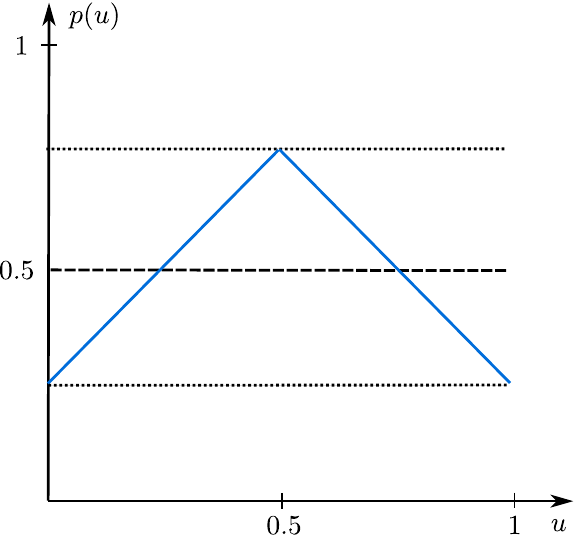}
\hspace{0.2mm}
\includegraphics[width=52mm]{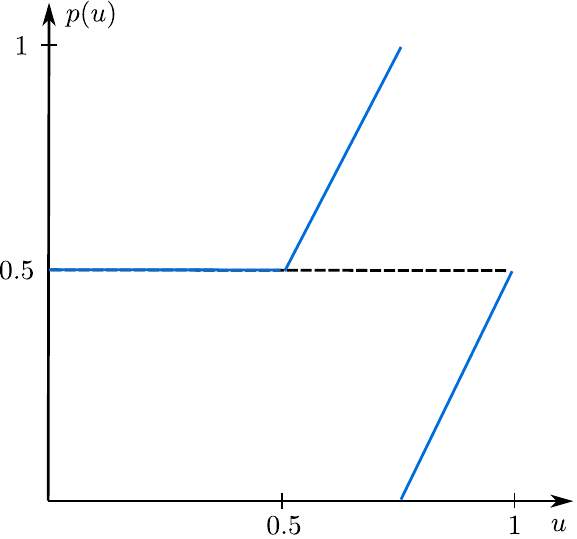}
\hspace{0.2mm}
\includegraphics[width=52mm]{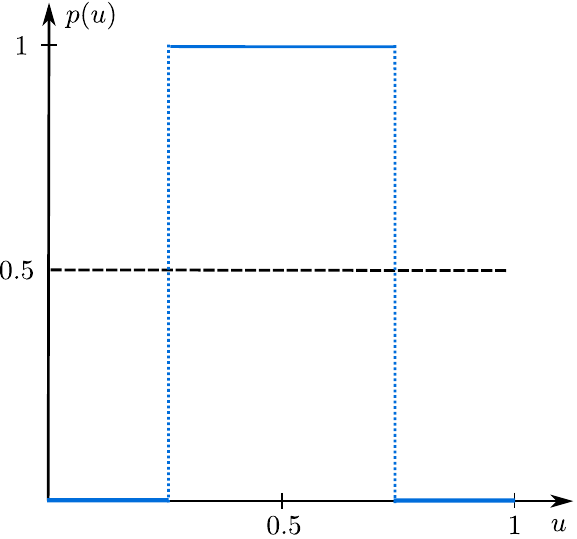}
\caption{Various propensity scores consistent with $\mathcal{T} = \{ 0.5 \}$-independence, when $\Prob(X=1) = 0.5$.}
\label{propensityScoresFig1}
\end{figure}

\begin{corollary}\label{corr:monotonicPropensityScores}
Suppose $X$ is binary and $U$ is continuously distributed; normalize $U \sim \text{Unif}[0,1]$. Suppose the propensity score $p$ is weakly monotonic and not constant on $(0,1)$. Then $U$ is not $\tau$-cdf independent of $X$ for all $\tau \in (0,1)$.
\end{corollary}

Corollary \ref{corr:monotonicPropensityScores} shows that a non-constant propensity score must be non-monotonic if it is to satisfy a $\tau$-cdf independence condition. This result can be extended as follows. Say that a function $f$ \emph{changes direction at least $K$ times} if there exists a partition of its domain into $K$ intervals such that $f$ is not monotonic on each interval.

\begin{corollary}\label{corr:SignChanges}
Suppose $X$ is binary and $U$ is continuously distributed; normalize $U \sim \text{Unif}[0,1]$. Suppose $U$ is $\mathcal{T}$-independent of $X$. Suppose there exists a version of $p$ without removable discontinuities. Partition $[0,1]$ by the sets $\mathcal{U}_k = [t_{k-1},t_k)$ for $k=1,\ldots,K-1$ with $t_0 = 0$, $t_K = 1$, and $\mathcal{U}_K = [t_{K-1},t_K]$ and such that for each $k$ there is a $\tau_k \in \mathcal{T}$ with $\tau_k \in \mathcal{U}_k$. Suppose $p$ is not constant over each set $\mathcal{U}_k$, $k=1,\ldots,K$. Then $p$ changes direction at least $K$ times.
\end{corollary}

This result essentially says that such propensity scores must oscillate up and down at least $K$ times (we assume $p$ does not have removable discontinuities to rule out trivial direction changes).  For example, as in figure \ref{propensityScoresFig1}, suppose we continue to have $\Prob(X=1) = 0.5$ but we add a few more isolated $\tau$'s to $\mathcal{T}$. Figure \ref{propensityScoresFig2} shows several propensity scores consistent with $\mathcal{T}$-independence for larger choices of $\mathcal{T}$. Consider the figure on the left, with $\mathcal{T} = \{ 0.25, 0.5, 0.75 \}$. Partition $[0,1] = [0,0.4) \cup [0.4,0.6) \cup [0.6,1]$. Then $p$ is not monotonic over each partition set, and each partition set contains one element of $\mathcal{T}$: $0.25 \in [0,0.4)$, $0.5 \in [0.4,0.6)$, and $0.75 \in [0.5,1]$. There are $K=3$ partition sets, and hence the corollary says $p$ must change direction at least 3 times. We see this in the figure since there are 3 interior local extrema. A similar analysis holds for the figure on the right. Overall, these triangular and sawtooth propensity scores illustrate the oscillation required by corollary \ref{corr:SignChanges}.

\begin{figure}[t]
\centering
\includegraphics[width=52mm]{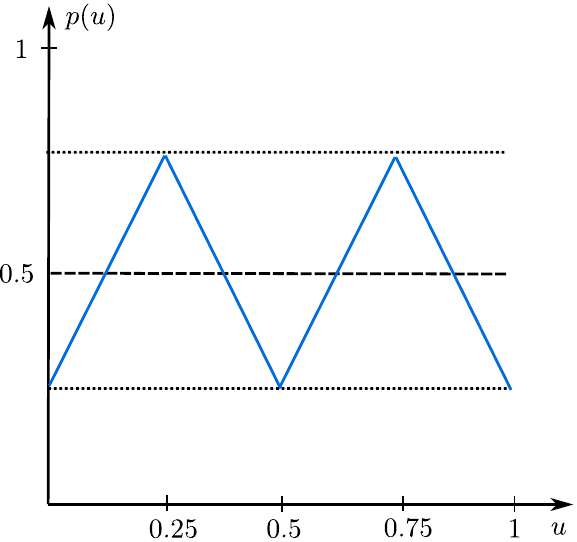}
\hspace{0.2mm}
\includegraphics[width=52mm]{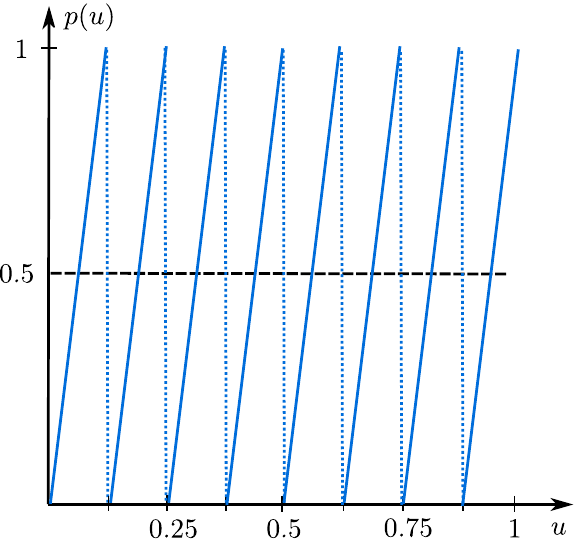}
\caption{Some propensity scores consistent with $\mathcal{T}$-independence when $\Prob(X=1) = 0.5$. Left: $\mathcal{T} = \{ 0.25, 0.5, 0.75 \}$. Right: $\mathcal{T} = \{ 0.125, 0.25, 0.375, 0.5, 0.625, 0.75, 0.875 \}$.}
\label{propensityScoresFig2}
\end{figure}

One final feature we document is that as long as there is some interval which is not in $\mathcal{T}$ then there is a propensity score which takes the most extreme values possible, 0 and 1.

\begin{corollary}\label{corr:TauIndep_ExtremeValues}
Suppose $X$ is binary and $U$ is continuously distributed; normalize $U \sim \text{Unif}[0,1]$. Suppose $[0,1] \setminus \mathcal{T}$ contains a non-degenerate interval. Then there exists a propensity score which is consistent with $\mathcal{T}$-independence of $U$ from $X$ and for which the sets
\[
	\{u \in [0,1]: p(u) = 0\}
	\qquad \text{and} \qquad
	\{u \in [0,1]: p(u) = 1\}
\]
have positive Lebesgue measure.
\end{corollary}

\subsection{Quantile Independence versus Mean Independence}

Like quantile independence, mean independence is commonly used to weaken statistical independence. For example, \cite{HeckmanIchimuraTodd1998} assume potential outcomes are mean independent of treatments, conditional on covariates. Our main result, theorem \ref{thm:AvgValueCharacterization}, allows us to compare the kinds of constraints on selection on unobservables imposed by quantile independence with the constraints imposed by mean independence. In this subsection, we briefly explore this comparison.

Say $U$ is \emph{mean independent} of $X$ if $\Exp(U \mid X=x) = \Exp(U)$ for all $x \in \supp(X)$. The following result follows immediately from this definition.

\begin{proposition}\label{prop:meanIndepConstraint}
Suppose $U$ is continuously distributed with density $f_U$. Suppose $X$ is binary with $\Prob(X=1) \in (0,1)$. Then $U$ is mean independent of $X$ if and only if
\begin{equation}\label{eq:meanIndepIntegral}
	\int_{-\infty}^\infty u \; f_U(u) \; p(u) \; du = \Exp(U)  \Prob(X=1).
\end{equation}
\end{proposition}

In particular, for comparison with theorem \ref{thm:AvgValueCharacterization}, if $U \sim \text{Unif}[0,1]$ then equation \eqref{eq:meanIndepIntegral} simplifies to
\[
	\int_0^1 2u \, p(u) \; du = \Prob(X=1).
\]
Theorem \ref{thm:AvgValueCharacterization} showed that quantile independence constrains the \emph{unweighted} average value of the latent propensity score over certain \emph{subintervals} of its domain. In contrast, proposition \ref{prop:meanIndepConstraint} shows that mean independence constrains a \emph{weighted} average value of the latent propensity score over its entire domain. Proposition \ref{prop:meanIndepConstraint} can be extended to multi-valued and continuous $X$ similar to our analysis of quantile independence in section \ref{sec:ctsX}; we omit this extension for brevity.

Although mean independence imposes different a constraint on the latent propensity score than quantile independence, it also requires non-constant latent propensity scores to be non-monotonic.

\begin{corollary}\label{corr:mean-indep prop scores}
Suppose $X$ is binary and $U$ is continuously distributed with finite mean and support equal to a possibly unbounded interval. Suppose the propensity score $p$ is weakly monotonic and not constant on the interior of its domain. Then $U$ is not mean independent of $X$.
\end{corollary}

\section{Multi-valued and Continuous Treatments}\label{sec:ctsX}

Thus far we have focused on binary $X$. In this section we extend our main characterization 
result (theorem \ref{thm:AvgValueCharacterization}) to both multi-valued discrete and continuous $X$. As in the binary $X$ case, our results show that the deviations from independence allowed by quantile independence require a kind of non-monotonic selection on unobservables.

We begin with the continuous case.

\begin{theorem}[Average value characterization]
\label{thm:AvgValueCharacterization_ctsX}
Suppose $X$ and $U$ are continuously distributed; normalize $U \sim \text{Unif}[0,1]$. Then $U$ is $\mathcal{T}$-independent of $X$ if and only if
\begin{equation}\label{eq:averageValueCondition_main_ctsX}
	\frac{1}{t_2 - t_1} \int_{t_1}^{t_2} \Prob(X > x \mid U = u) \; du = \Prob(X > x)
	\quad \text{for all $x \in \supp(X)$}
\end{equation}
for all $t_1, t_2 \in\mathcal{T} \cup \{ 0,1 \}$ with $t_1 < t_2$.
\end{theorem}

The interpretation of equation \eqref{eq:averageValueCondition_main_ctsX} is similar to the binary $X$ case: $\mathcal{T}$-independence holds if and only if, for each possible level of treatment $x$, and for each interval with endpoints in $\mathcal{T} \cup \{ 0,1 \}$ the average value of the conditional probability of receiving treatment larger than $x$ given the unobservable equals the overall unconditional probability of receiving treatment larger than $x$. Notice that, by adding $-1$ to each side of equation \eqref{eq:averageValueCondition_main_ctsX}, this constraint can equivalently be seen as a constraint on the conditional cdf $F_{X \mid U}$.

As in the binary $X$ case, the constraint \eqref{eq:averageValueCondition_main_ctsX} imposes a non-monotonicity condition.

\begin{corollary}\label{corr:noRegDependence}
Suppose $X$ and $U$ are continuously distributed. Suppose there is some $x \in \supp(X)$ such that $\Prob(X > x \mid U=u)$ is weakly monotonic and not constant over $u \in \text{int}[\supp(U)]$. Then $U$ is not $\tau$-cdf independent of $X$ for all $\tau \in \supp(U)$.
\end{corollary}

For example, suppose $X$ is level of education, $x$ is completing college, and $U$ is ability. Then any nontrivial $\mathcal{T}$-independence condition implies that at some point increasing ability lowers the probability of getting more than a college education.

The monotonicity condition of corollary \ref{corr:noRegDependence} dates back to \cite{Tukey1958} and \cite{Lehmann1966}, who give the following definition.

\begin{definition}
Say $X$ is \emph{positively [negatively] regression dependent} on $U$ if $\Prob(X > x \mid U=u)$ is weakly increasing [weakly decreasing] in $u$, for all $x \in \R$. Say $X$ is \emph{regression dependent} on $U$ if it is either positively or negatively regression dependent on $U$.
\end{definition}

Thus corollary \ref{corr:noRegDependence} states that we cannot simultaneously have quantile independence of $U$ on $X$ and regression dependence of $X$ on $U$ (except when $X \independent U$).

\cite{LehmannRomano2005} call positive regression dependence an ``intuitive meaning of positive dependence''. To support this claim, \cite{Lehmann1966} gave the following simple sufficient conditions for regression dependence: If one can write $X = \pi_0 + \pi_1 U + V$ where $\pi_0$ and $\pi_1$ are constants and $V$ is a random variable independent of $U$, then $X$ is regression dependent on $U$ if $\pi_1 \neq 0$. In particular, if $X$ and $U$ are jointly normally distributed then they are regression dependent so long as they have nonzero correlation. While these are special cases, theorem 5.2.10 on page 196 of \cite{Nelsen2006} provides a general characterization of regression dependence in terms of the copula between $X$ and $U$, when both variables are continuous. In particular, if $C_{X,U}(x,u)$ is the copula for $(X,U)$, $X$ is regression dependent on $U$ if and only if $C_{X,U}(x,\cdot)$ is concave for any $x \in [0,1]$.

Regression dependence is also known as \emph{stochastic monotonicity}, since it is equivalent to the set of cdfs $\{ F_{X \mid U}(\cdot \mid u) : u \in \supp(U) \}$ being either increasing or decreasing in the first order stochastic dominance ordering. A large literature studies tests of stochastic monotonicity. For example, \cite{LeeLintonWhang2009} state that stochastic monotonicity is ``of interest in many applications in economics'', and provide many references which use such monotonicity assumptions. See \cite{DelgadoEscanciano2012}, \cite{HsuLiuShi2016}, and \cite{Seo2017} for further work on testing stochastic monotonicity.
In a different application of stochastic monotonicity, \cite{BlundellEtAl2007} study the classic problem of identifying the distribution of potential wages, given that wages are only observed for workers. Following \cite{ManskiPepper2000}, they argue that stochastic monotonicity assumptions are often plausible. They specifically consider stochastic monotonicity of wages on labor force participation status, as well as stochastic monotonicity of wages on an instrument. They furthermore provide a detailed analysis of when stochastic monotonicity assumptions may not be plausible. 

Corollary \ref{corr:noRegDependence} shows that any assumption of $\mathcal{T}$-independence of $U$ on $X$ rules out stochastic monotonicity of $X$ on $U$. Thus, if one wants to allow for a class of deviations from independence which includes stochastically monotonic selection, assumptions of quantile independence of $U$ on $X$ should not be used. Conversely, if one makes a quantile independence assumption of $U$ on $X$, one should argue why stochastically non-monotonic selection models are the deviations of interest. We discuss these issues further in section \ref{sec:LatentSelectionModels}.

The following theorem extends our characterization results to the discrete $X$ case.

\begin{theorem}[Average value characterization]\label{thm:AvgValueCharacterization_discreteX}
Suppose $X$ is discrete with support $\{ x_1,\ldots,x_K \}$ and $\Prob(X=x_k) \in (0,1)$ for $k \in \{ 1,\ldots,K \}$. Suppose $U$ is continuously distributed; normalize $U \sim \text{Unif}[0,1]$. Then $U$ is $\mathcal{T}$-independent of $X$ if and only if
\begin{equation}\label{eq:averageValueCondition_main_discreteX}
	\frac{1}{t_2 - t_1} \int_{[t_1,t_2]} \Prob(X=x_k \mid U=u) \; du = \Prob(X=x_k)
	\qquad \text{for all $k \in \{ 1,\ldots,K \}$}
\end{equation}
for all $t_1, t_2 \in\mathcal{T} \cup \{ 0,1 \} $ with $t_1 < t_2$.
\end{theorem}

This result has a similar interpretation as our previous results for binary $X$ and continuous $X$. First, we have the following corollary.

\begin{corollary}\label{corr:discreteXnotMonotone}
Suppose $X$ is discrete with support $\{ x_1,\ldots,x_K \}$. Suppose $U$ is continuously distributed. Suppose there is some $k \in \{ 1,\ldots,K \}$ such that $\Prob(X \geq x_k \mid U=u)$ is weakly monotonic and not constant over $u \in \text{int}[\supp(U)]$. Then $U$ is not $\tau$-cdf independent of $X$ for all $\tau \in \supp(U)$.
\end{corollary}

The interpretation is analogous to corollary \ref{corr:noRegDependence}. Second, all of the interpretations given in section \ref{sec:characterizationSection} apply to the probabilities $\Prob(X = x_k \mid U=u)$ for $k \in \{1,\ldots,K \}$. In particular, these conditional probabilities must be non-monotone. This result is primarily relevant for the lowest treatment level ($k=1$) and the highest treatment level ($k=K$), since non-monotonicity of the middle probabilities would be implied, for example, by a simple ordered threshold crossing model, like $X = x_k$ if $\alpha_k \leq U \leq \alpha_{k+1}$ for constants $\alpha_k \leq \alpha_{k+1}$, $k \in \{ 1,\ldots,K \}$.

\section{Latent Selection Models}\label{sec:LatentSelectionModels}

Many econometric models obtain point identification via quantile independence restrictions, rather than statistical independence. These results are often motivated solely by the fact that quantile independence is weaker than statistical independence, and hence results using only quantile independence are `more robust' than those using statistical independence. As we have emphasized, any deviation from statistical independence of $U$ and $X$ allows for certain forms of treatment selection, in the sense that the distribution of $X \mid U=u$ depends nontrivially on $u$. Thus the choice of an assumption weaker than statistical independence depends on the class of deviations one wishes to be robust against. Since this class is often not explicitly specified, we refer to such deviations as \emph{latent selection models}. Our main results in sections \ref{sec:characterizationSection} and \ref{sec:ctsX} characterize the set of latent selection models allowed by quantile independence restrictions.

In this section, we study two standard econometric selection models. We discuss different assumptions on the economic primitives which lead these models to be either consistent or inconsistent with quantile independence restrictions. We only consider single-agent models, but similar analyses can likely be done for multi-agent models.

\subsection{A Model for Continuous Treatment Choice}

We first consider a selection model discussed by \citet[pages 1484--1485]{ImbensNewey2009}.\footnote{\citet[section 4.2]{Pakes1994} and \citet[pages 283--285]{BlundellMatzkin2014} give additional examples of economic selection models which yield strict monotonicity in a first stage unobservable.} Consider a population of people deciding how much education to obtain. Let $Y$ denote earnings, $X$ the chosen level of education, and $U$ ability. Let $m(x,u)$ denote the earnings production function. Hence earnings, education, and ability jointly satisfy
\[
	Y = m(X,U).
\]
Let $c(x,z)$ denote the cost of obtaining education level $x$. $Z$ denotes a variable which shifts cost and is known to agents. Suppose agents do not perfectly know their own ability, but instead observe a noisy signal $V$ of $U$. Agents know the joint distribution of $(U,V)$. Suppose agents choose $X$ to maximize expected earnings, minus costs. Thus agents solve the problem
\begin{equation}\label{ImbensNeweySelectionEquation}
	\max_{x \geq 0} \big( \Exp[ m(x,U) \mid V=v,Z=z] - c(x,z) \big).
\end{equation}
The following proposition is a variation of a result stated by \cite{ImbensNewey2009}.

\begin{proposition}\label{prop:ImbensNeweySelection}
Consider the selection model \eqref{ImbensNeweySelectionEquation}. Suppose the following hold.
\begin{enumerate}
\item $m$ is twice continuously differentiable in both components, with
\[
	\frac{\partial^2 m(x,u)}{\partial x^2} < 0
	\quad \text{and} \quad
	\frac{\partial^2 m(x,u)}{\partial x \partial u} > 0
\]
for all $x > 0$ and $u \in \supp(U)$. $m(x,u)$ is strictly increasing in both $x$ and $u$. For each $u$, both the first and second derivatives of $m$ in $x$ are bounded in absolute value by a function $B(x,u)$ with $\Exp[ B(x,U) \mid V=v ] < \infty$ for all $v \in \supp(V)$ and all $x > 0$. For each $u \in \supp(U)$, $\partial m(x,u) / \partial x \rightarrow 0$ as $x \nearrow \infty$ and $\partial m(x,u) / \partial x \rightarrow \infty$ as $x \searrow 0$.

\item $c(x,z)$ is twice continuously differentiable in $x$ for each $z \in \supp(Z)$ with
\[
	\frac{\partial c(x,z)}{\partial x} > 0
	\quad \text{and} \quad
	\frac{\partial^2 c(x,z)}{\partial x^2} > 0
\]
for all $x > 0$ and $z \in \supp(Z)$. For each $z \in \supp(Z)$, $\partial c(x,z) / \partial x \rightarrow \infty$ as $x \nearrow \infty$ and $\partial c(x,z) / \partial x \rightarrow 0$ as $x \searrow 0$.

\item $Z \independent (U,V)$.

\item $(U,V)$ are jointly continuously distributed. $U$ and $V$ satisfy the strict monotone likelihood ratio property (MLRP).
\end{enumerate}
Then for each $(z,v) \in \supp(Z,V)$ there is a unique solution to the problem \eqref{ImbensNeweySelectionEquation}, which we denote by $h(z,v)$. Moreover, $h(z,\cdot)$ is strictly increasing for each $z \in \supp(Z)$.
\end{proposition}

Assumption 1 constrains the earnings production function. Higher education increases earnings, with diminishing marginal returns. Higher ability also increases earnings. Importantly, earnings and ability are complementary. As \cite{ImbensNewey2009} mention, a Cobb-Douglas production function satisfies these assumptions. Assumption 2 constrains the cost function. Cost is increasing in earnings with increasing marginal cost. Assumption 3 implies that $Z$ has no information about agents' true ability $U$. Assumption 4 formalizes the idea that $V$ is a signal of $U$. See \cite{Milgrom1981} for the definition and further discussion of the strict MLRP. The strict MLRP is also sometimes called \emph{strict affiliation}. \cite{AtheyHaile2002,AtheyHaile2007} and \cite{PinkseTan2005} discuss strict affiliation in the context of auction models. The MLRP, and hence the strict MLRP, implies that $V$ is positive regression dependent on $U$.

Proposition \ref{prop:ImbensNeweySelection} gives conditions under which the treatment selection equation has the form
\[
	X = h(Z,V)
\]
where $h(z,\cdot)$ is strictly increasing for each $z \in \supp(Z)$. This is a common restriction imposed in the control function literature. The following proposition shows that this monotonicity restriction combined with regression dependence of the signal $V$ on ability $U$ implies regression dependence of chosen education $X$ on ability $U$.

\begin{proposition}\label{corr:ImbensNeweySelection}
Suppose $X = h(Z,V)$ where
\begin{enumerate}
\item $h(z,\cdot)$ is strictly monotone for each $z \in \supp(Z)$.

\item $Z \independent (U,V)$.

\item $V$ is continuously distributed and is regression dependent on $U$.
\end{enumerate}
Then $X$ is regression dependent on $U$.
\end{proposition}

\begin{corollary}\label{corr:ImbensNeweySelectionFinalResult}
Consider the selection model \eqref{ImbensNeweySelectionEquation}. Suppose the assumptions of proposition \ref{prop:ImbensNeweySelection} hold. Then no quantile independence conditions of $U$ on $X$ can hold.
\end{corollary}

This corollary is perhaps not surprising, since one would not typically consider $X$ to be `exogenous' in the model above. Indeed, \cite{ImbensNewey2009} go on to assume that $Z$ is observable and then use its variation to identify treatment effects. To reiterate our previous points, however: A quantile independence assumption allows for selection on unobservables, since it is weaker than full independence. Proposition \ref{corr:ImbensNeweySelection} shows that the form of this allowed selection is not compatible with the selection model described above. On the other hand, if one of the assumptions of proposition \ref{prop:ImbensNeweySelection} fails, then $X$ might not be regression dependent on $U$, and hence a quantile independence condition might hold. In particular, the assumption that $V$ and $U$ satisfy the strict MLRP (which implies regression dependence of $V$ on $U$) could perhaps be dropped. Researchers using quantile independence assumptions should argue why the class of selection models compatible with the quantile independence conditions---as specified in our characterization theorems---are the deviations of interest.

\subsection{The Roy Model of Binary Treatment Choice}

Let $X \in \{ 0,1 \}$ be a binary treatment and $Y_1$ and $Y_0$ denote potential outcomes. We study identification of this model in section \ref{sec:treatmentEffects}. Here we study the class of latent selection models consistent with quantile independence. Suppose agents choose treatment to maximize their outcome:
\begin{equation}\label{eq:RoyModel}
	X = \indicator(Y_1 > Y_0).
\end{equation}
This is the classical Roy model (see \citealt{HeckmanVytlacil2007}). Suppose we are interested in identifying treatment on the treated parameters. Then identification depends on our assumptions about the stochastic relationship between $X$ and $Y_0$. In particular, one might consider assuming that some quantile of $Y_0$ is independent of $X$. As we have discussed, such an assumption constrains the latent selection model of $X$ given $Y_0$. Specifically, consider the latent propensity score,
\begin{align*}
	p(y_0)
		&\equiv \Prob(X=1 \mid Y_0 = y_0) \\
		&= \Prob(Y_1 > y_0 \mid Y_0 = y_0).
\end{align*}
The second line follows by our Roy model treatment choice assumption. Thus regression dependence of $Y_1$ on $Y_0$ implies that $p$ is monotonic and hence no quantile independence conditions of $Y_0$ on $X$ can hold (except when $Y_1 \independent Y_0$ or if $X$ is degenerate, as when treatment effects $Y_1-Y_0$ are constant). In particular, any quantile independence condition of $Y_0$ on $X$ rules out bivariate normally distributed $(Y_1,Y_0)$, unless $Y_1 \independent Y_0$. 

There are, however, joint distributions of $(Y_1,Y_0)$ such that $Y_1$ is not regression dependent on $Y_0$. For example, let $Y_0 \sim \normal(0,1)$ and $Y_1 = Y_0 + \mu(Y_0) - \varepsilon$ where $\mu$ is a deterministic function and $\varepsilon \sim \normal(0,1)$, $\varepsilon \independent Y_0$. Then
\[
	p(y_0) = \Prob(X=1 \mid Y_0 = y_0)
	= \Phi[ \mu( y_0 ) ],
\]
where $\Phi$ is the standard normal cdf. If $\mu$ is non-monotonic then $p$ will also be non-monotonic. For this joint distribution of potential outcomes, the unit level treatment effects $Y_1-Y_0$ conditional on the baseline outcome $Y_0=y_0$ are distributed $\normal( \mu(y_0), 1)$. Hence non-monotonicity of $\mu$ implies that the mean of this distribution of treatment effects is not monotonic. For instance, suppose the outcome is earnings and treatment is completing college. Suppose
\begin{align*}
	\mu(y_0) &> 0 \qquad\text{if $y_0 \in (\alpha,\beta)$} \\
	\mu(y_0) &\leq 0 \qquad\text{if $y_0 \in (-\infty,\alpha] \cup [\beta,\infty)$}
\end{align*}
for $-\infty < \alpha < \beta < \infty$. Then people with sufficiently small or sufficiently large earnings when they do not complete college do not benefit from completing college, on average. People with moderate earnings when they do not complete college, on the other hand, do typically benefit from completing college. Put differently, if potential earnings is an increasing deterministic function of ability, then low and high ability people do not benefit from completing college; only middle ability people do. This kind of joint distribution of potential outcomes combined with the Roy model assumption \eqref{eq:RoyModel} on treatment selection produce non-monotonic latent propensity scores.

While that is just one example joint distribution of $(Y_1,Y_0)$ where regression dependence fails, theorem 5.2.10 on page 196 of \cite{Nelsen2006} characterizes the set of copulas for which $Y_1$ is regression dependent on $Y_0$, when both are continuously distributed. This result therefore also tells us the set of copulas where $Y_1$ is \emph{not} regression dependent on $Y_0$. Among these copulas, $\mathcal{T}$-independence of $Y_0$ from $X$ will specify a further subset of allowed dependence structures. The precise set is given by all copulas which lead to latent propensity scores that satisfy the average value constraint. One could conversely pick a set of allowed copulas and use theorem \ref{thm:AvgValueCharacterization} to obtain a set of quantile independence conditions that might hold. This would allow one to obtain an identified set for parameters like the average treatment effect for the treated under the given constraints on the set of copulas, although we do not pursue this here.

As in the continuous treatment case, researchers using quantile independence assumptions should argue that the set of primitives---like the joint distributions of $(Y_1,Y_0)$ in the Roy model---allowed by quantile independence are the deviations of interest. In other cases, quantile independence may be implausible, as in Heckman, Smith, and Clements' \citeyearpar{HeckmanSmithClements1997} empirical analysis of the Job Training Partnership Act (JTPA), who find that ``plausible impact distributions require high measures of positive dependence [of $Y_1$ on $Y_0$]'' (page 506).

\section{The Identifying Power of the Average Value Constraint}\label{sec:treatmentEffects}

Our main result in section \ref{sec:characterizationSection} shows that quantile independence imposes a constraint on the average value of a latent propensity score. In this section, we study the implications of this result for identification. We first use our characterization to motivate an assumption weaker than quantile independence, which we call \emph{$\mathcal{U}$-independence}. The only difference between these two assumptions is that quantile independence imposes the average value constraint while $\mathcal{U}$-independence does not. Hence the difference between identified sets obtained under these two assumptions is a measure of the identifying power of the average value constraint, which is the feature of quantile independence that requires the latent propensity score to be non-monotonic.

To compute such identified sets and perform such a comparison, one must first specify an econometric model and a parameter of interest. While this can be done in many different models, we focus on a simple but important model: the standard potential outcomes model of treatment effects with a binary treatment. We adapt the analysis of \cite{MastenPoirier2017} to derive identified sets for the average treatment effect for the treated (ATT) and the quantile treatment effect for the treated (QTT) under both $\mathcal{T}$- and $\mathcal{U}$-independence. We then compare these identified sets in a numerical illustration. In this illustration, the identified sets are significantly larger under $\mathcal{U}$-independence, implying that the average value constraint has substantial identifying power.

\subsection{Weakening Quantile Independence}

Throughout this section, we focus on the case where $\mathcal{T}$ is an interval. In this case, we show that latent propensity scores consistent with $\mathcal{T}$-independence have two features: (a) they are flat on $\mathcal{T}$ and (b) they are non-monotonic outside the flat regions, such that the average value constraint \eqref{eq:averageValueCondition_main} is satisfied. We use this finding to motivate a weaker assumption which retains feature (a) but drops feature (b). We call this assumption \emph{$\mathcal{U}$-independence}. While one may consider this a reasonable assumption, our primary motivation for studying $\mathcal{U}$-independence is as a tool for understanding quantile independence.

We begin with the following corollary to theorem \ref{thm:AvgValueCharacterization}.

\begin{corollary}\label{prop:TauImpliesU}
Suppose $X$ is binary and $U$ is continuously distributed; normalize $U \sim \text{Unif}[0,1]$. Let $\mathcal{T} = [a,b]\subseteq [0,1]$. Then $\mathcal{T}$-independence of $U$ from $X$ implies
\begin{equation}\label{eq:UindepEq}
	\Prob(X=1 \mid U=u) = \Prob(X=1)
\end{equation}
for almost all $u \in \mathcal{T}$.
\end{corollary}

Corollary \ref{prop:TauImpliesU} shows that $\mathcal{T}$-independence requires the latent propensity score to be flat on $\mathcal{T}$ and equal to the overall unconditional probability of being treated. The first property---that the latent propensity score is flat on $\mathcal{T}$---means that random assignment holds within the subpopulation of units whose unobservables are in the set $\mathcal{T}$; that is, $X \independent U \mid \{ U \in \mathcal{T} \}$. Corollary \ref{prop:TauImpliesU} can be generalized to allow $\mathcal{T}$ to be a finite union of intervals, but we omit this for simplicity.

This corollary motivates the following definition.

\begin{definition}
Let $\mathcal{U} \subseteq \supp(U)$. Say that $U$ is \emph{$\mathcal{U}$-independent} of $X$ if equation \eqref{eq:UindepEq} holds for almost all $u \in \mathcal{U}$.
\end{definition}

Corollary \ref{prop:TauImpliesU} shows that $\mathcal{T}$-independence implies $\mathcal{U}$-independence with $\mathcal{T}=\mathcal{U}$. The converse does not hold since $\mathcal{T}$-independence furthermore requires the average value constraint to hold, by theorem \ref{thm:AvgValueCharacterization}. For example, figure \ref{TauIndepFlatPropScore} shows two latent propensity scores. One satisfies $\mathcal{T}$-independence, but the other only satisfies $\mathcal{U}$-independence. Finally, note that $\mathcal{U}$-independence is a nontrivial assumption only when $\Prob(U \in \mathcal{U}) > 0$. Conversely, $\mathcal{T}$-independence is nontrivial even when $\mathcal{T}$ is a singleton.

\begin{figure}[t]
\centering
\includegraphics[width=52mm]{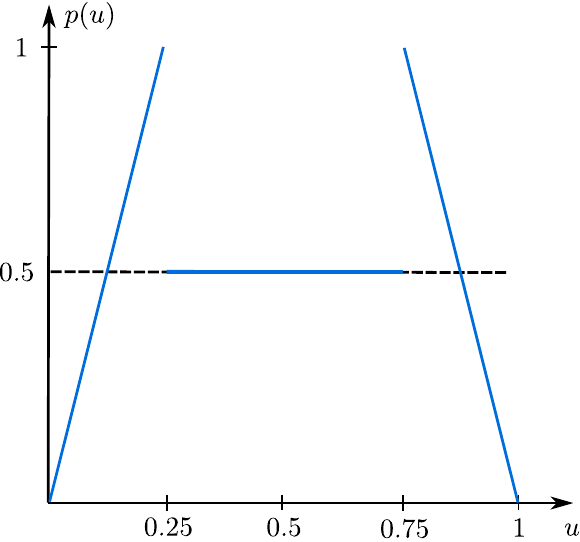}
\hspace{5mm}
\includegraphics[width=52mm]{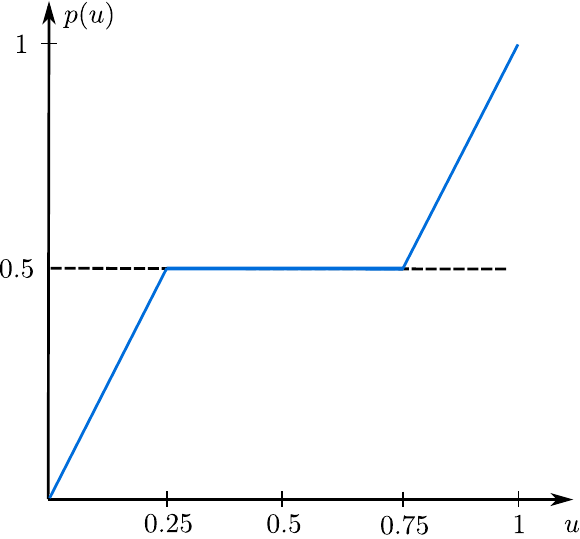}
\caption{Let $\mathcal{T} = \mathcal{U} = [0.25,0.75]$ and $\Prob(X=1) = 0.5$. This figure shows a propensity score consistent with both $\mathcal{T}$- and $\mathcal{U}$-independence on the left and a propensity score consistent with $\mathcal{U}$-, but not $\mathcal{T}$-independence on the right.}
\label{TauIndepFlatPropScore}
\end{figure}

The following result extends corollary \ref{prop:TauImpliesU} to allow $X$ to be multi-valued or continuous.

\begin{corollary}\label{prop:TauImpliesU_ctsX}
Suppose $U$ is continuously distributed; normalize $U \sim \text{Unif}[0,1]$. Let $\mathcal{T} = [a,b]\subseteq [0,1]$. Then $\mathcal{T}$-independence of $U$ from $X$ implies $F_{X \mid U}(x \mid u) = F_X(x)$ for all $x \in \R$ and almost all $u \in \mathcal{T}$.
\end{corollary}

Although we focus on binary $X$ for the remainder of this section, this corollary suggests that we can generalize our definition of $\mathcal{U}$-independence to allow multi-valued or continuous treatments by specifying $F_{X \mid U}(x \mid u) = F_X(x)$ for all $x \in \R$ and almost all $u \in \mathcal{U}$.

\subsection{The Potential Outcomes Model}

Let $Y_1$ and $Y_0$ denote unobserved potential outcomes. Let $X \in \{ 0,1\}$ be an observed binary treatment. We observe the scalar outcome variable
\begin{equation}\label{eq:potential outcomes}
	Y = X Y_1 + (1-X) Y_0.
\end{equation}
All results hold if we condition on an additional vector of observed covariates $W$. For simplicity we omit these covariates.
Let $p_x = \Prob(X=x)$ for $x \in \{0,1\}$. We impose the following assumption on the joint distribution of $(Y_1,Y_0,X)$.

\setcounter{partialIndepSection}{1}
\begin{partialIndepAssump}\label{assn:continuity}
For each $x,x' \in \{0,1\}$:
\begin{enumerate}
\item \label{A1_1} $Y_x \mid X=x'$ has a strictly increasing and  continuous distribution function on its support, $\supp(Y_x \mid X=x')$.

\item \label{A1_3} $\supp(Y_x \mid X=x') = \supp(Y_x) = [\underline{y}_x,\overline{y}_x]$ where $-\infty \leq \underline{y}_x < \overline{y}_x \leq \infty$.

\item \label{A1_4} $p_x > 0$.
\end{enumerate}
\end{partialIndepAssump}

Via A\ref{assn:continuity}.\ref{A1_1}, we restrict attention to continuously distributed potential outcomes. A\ref{assn:continuity}.\ref{A1_3} states that the unconditional and conditional supports of $Y_x$ are equal, and are a possibly infinite closed interval. This assumption implies that the endpoints $\underline{y}_x$ and $\overline{y}_x$ are point identified. We maintain A\ref{assn:continuity}.\ref{A1_3} for simplicity, but it can be relaxed using similar derivations as in \cite{MastenPoirier2016}. A\ref{assn:continuity}.\ref{A1_4} is an overlap assumption.

We focus on two parameters: The average treatment effect for the treated,
\[
	\text{ATT} = \Exp(Y_1 - Y_0 \mid X=1),
\]
and the quantile treatment effect for the treated,
\[
	\text{QTT}(q) = Q_{Y_1 \mid X}(q \mid 1) - Q_{Y_0 \mid X}(q \mid 1),
\]
for $q \in (0,1)$. Treatment on the treated parameters are particularly simple to analyze since the distribution of $Y_1 \mid X=1$ is point identified directly from the observed distribution of $Y \mid X=1$. Hence we only need to make assumptions on the relationship between $Y_0$ and $X$. Our analysis can be extended to parameters like ATE by imposing $\mathcal{T}$- or $\mathcal{U}$-independence between $Y_1$ and $X$ as well as between $Y_0$ and $X$.

\subsection{Partial Identification of Treatment Effects}\label{sec:treatmentEffectBounds}

In this subsection we derive sharp bounds on the ATT and $\text{QTT}(q)$ under both $\mathcal{T}$- and $\mathcal{U}$-independence. Since $F_{Y_1 \mid X}(\cdot \mid 1) = F_{Y \mid X}(\cdot \mid 1)$,
\[
	\text{QTT}(q) = Q_{Y \mid X}(q \mid 1) - Q_{Y_0 \mid X}(q \mid 1).
\]
Hence it suffices to derive bounds on $Q_{Y_0 \mid X}(q \mid 1)$. Let $0 < a \leq b < 1$. We define the functions $\overline{Q}_{Y_0 \mid X}^\mathcal{T}(q \mid 1)$, $\underline{Q}_{Y_0 \mid X}^\mathcal{T}(q \mid 1)$, $\overline{Q}_{Y_0 \mid X}^\mathcal{U}(q \mid 1)$, and $\underline{Q}_{Y_0 \mid X}^\mathcal{U}(q \mid 1)$ in appendix \ref{appendix:cdfbounds}. These are piecewise functions which depend on $a$, $b$, $p_1$, and $Q_{Y \mid X}$.

\begin{proposition}\label{prop:quantileTrtBounds}
Let A\ref{assn:continuity} hold. Suppose $Y_0$ is $\mathcal{T}$-independent of $X$ with $\mathcal{T} = [Q_{Y_0}(a),Q_{Y_0}(b)]$, $0 < a \leq b < 1$. Suppose the joint distribution of $(Y,X)$ is known. Let $q \in (0,1)$. Then $Q_{Y_0 \mid X}(q \mid 1)$ lies in the set
\begin{equation}\label{eq:Y0quantileIdentifiedSet}
	\left[ \underline{Q}_{Y_0 \mid X}^\mathcal{T}(q \mid 1), \, \overline{Q}_{Y_0 \mid X}^\mathcal{T}(q \mid 1) \right].
\end{equation}
Moreover, the interior of this set is sharp. Finally, the proposition continues to hold if we replace $\mathcal{T}$ with $\mathcal{U}$.
\end{proposition}

$\mathcal{T}$-independence of $Y_0$ from $X$ with $\mathcal{T} = [Q_{Y_0}(a),Q_{Y_0}(b)]$ is equivalent to the quantile independence assumptions $Q_{Y_0 \mid X}(\tau \mid x) = Q_{Y_0}(\tau)$ for all $\tau \in [a,b]$, by A\ref{assn:continuity}. The bounds \eqref{eq:Y0quantileIdentifiedSet} are also sharp for the function $Q_{Y_0 \mid X}(\cdot \mid 1)$ in a sense similar to that used in proposition \ref{prop:TauCDFbounds} in appendix \ref{sec:cdfIdent}; we omit the formal statement for brevity. This functional sharpness delivers the following result.

\begin{corollary}\label{corr:ATT_Y0bounds}
Suppose the assumptions of proposition \ref{prop:quantileTrtBounds} hold. Suppose $\Exp( | Y | \mid X=x) < \infty$ for $x \in \{0,1\}$. Then $\Exp(Y_0 \mid X=1)$ lies in the set
\[
	\left[ \underline{\Exp}^\mathcal{T}(Y_0 \mid X=1), \,\overline{\Exp}^\mathcal{T}(Y_0 \mid X=1) \right]
	\equiv
	\left[
		\int_0^1 \underline{Q}_{Y_0 \mid X}^\mathcal{T}(q \mid 1) \; dq, \,
		\int_0^1 \overline{Q}_{Y_0 \mid X}^\mathcal{T}(q \mid 1) \; dq
	\right].
\]
Moreover, the interior of this set is sharp. Finally, the corollary continues to hold if we replace $\mathcal{T}$ with $\mathcal{U}$.
\end{corollary}

By proposition \ref{prop:quantileTrtBounds} we have that $\mathcal{T}$-independence implies that $\text{QTT}(q)$ lies in the set
\[
	\left[ Q_{Y \mid X}(q \mid 1) - \overline{Q}_{Y_0 \mid X}^\mathcal{T}(q \mid 1), \;
		Q_{Y \mid X}(q \mid 1) - \underline{Q}_{Y_0 \mid X}^\mathcal{T}(q \mid 1) \right]
\]
and that the interior of this set is sharp. Likewise for $\mathcal{U}$-independence. If $q \in \mathcal{T}$, then $\text{QTT}(q)$ is point identified under $\mathcal{T}$-independence (as is immediate from our bound expressions in appendix \ref{appendix:cdfbounds}). This result---that a single quantile independence condition can be sufficient for point identifying a treatment effect---was shown by \cite{Chesher2003}. A similar result holds in the instrumental variables model of \cite{ChernozhukovHansen2005} and the LATE model of \cite{ImbensAngrist1994}. See the discussion around assumption 4 in section 1.4.3 of \cite{MellyWuthrich2017}.

By corollary \ref{corr:ATT_Y0bounds} we have that $\mathcal{T}$-independence implies that ATT lies in the set
\[
	\left[ \Exp(Y \mid X=1) - \overline{\Exp}^\mathcal{T}(Y_0 \mid X=1), \;
		\Exp(Y \mid X=1) - \underline{\Exp}^\mathcal{T}(Y_0 \mid X=1) \right]
\]
and that the interior of this set is sharp. Likewise for $\mathcal{U}$-independence. Furthermore, in appendix \ref{appendix:cdfbounds} we show that these ATT bounds have simple analytical expressions, obtained from integrating our closed form expressions for the bounds on $Q_{Y_0 \mid X}(q \mid 1)$.

\subsection{Numerical Illustration}

By corollary \ref{prop:TauImpliesU}, $\mathcal{T}$-independence implies $\mathcal{U}$-independence for $\mathcal{T} = \mathcal{U}$. Hence identified sets using $\mathcal{T}$-independence must necessarily be weakly contained within identified sets using only $\mathcal{U}$-independence, when $\mathcal{T} = \mathcal{U}$. In this subsection, we use a numerical illustration to explore the magnitude of this difference. Since $\mathcal{T}$-independence is simply $\mathcal{U}$-independence combined with the average value constraint, the size difference between these identified sets tells us the identifying power of the average value constraint.

For $x \in \{0,1 \}$, suppose the density of $Y \mid X = x$ is
\[
	f_{Y \mid X}(y \mid x) = \frac{1}{\gamma x + 1} \phi_{[-4,4]} \left( \frac{y - \pi x}{\gamma x + 1} \right)
\]
where $\phi_{[-4,4]}$ is the pdf for the truncated standard normal on $[-4,4]$. $X$ is binary with $\Prob(X=1) = 0.5$. Let $\gamma = 0.1$ and $\pi = 1$.

Under the full independence assumption $Y_0 \independent X$, this dgp implies that treatment effects are heterogeneous, with an average treatment effect for the treated of $\text{ATT} = \pi = 1$. The quantile treatment effect for the treated at $q = 0.5$ also equals $\pi = 1$ under full independence. We continuously relax full independence by considering $\mathcal{T}$- and $\mathcal{U}$-independence with the choice $\mathcal{T} = \mathcal{U} = [\delta,1-\delta]$ for $\delta \in [0,0.5]$. For $\delta = 0$, this choice corresponds to full independence under both classes of assumptions. For $\delta = 0.5$, this choice corresponds to median independence for $\mathcal{T}$-independence, and no assumptions for $\mathcal{U}$-independence. Values of $\delta$ between 0 and $0.5$ yield partial independence between $Y_0$ and $X$ for both classes of assumptions.

\begin{figure}[t]
\centering
\includegraphics[width=80mm]{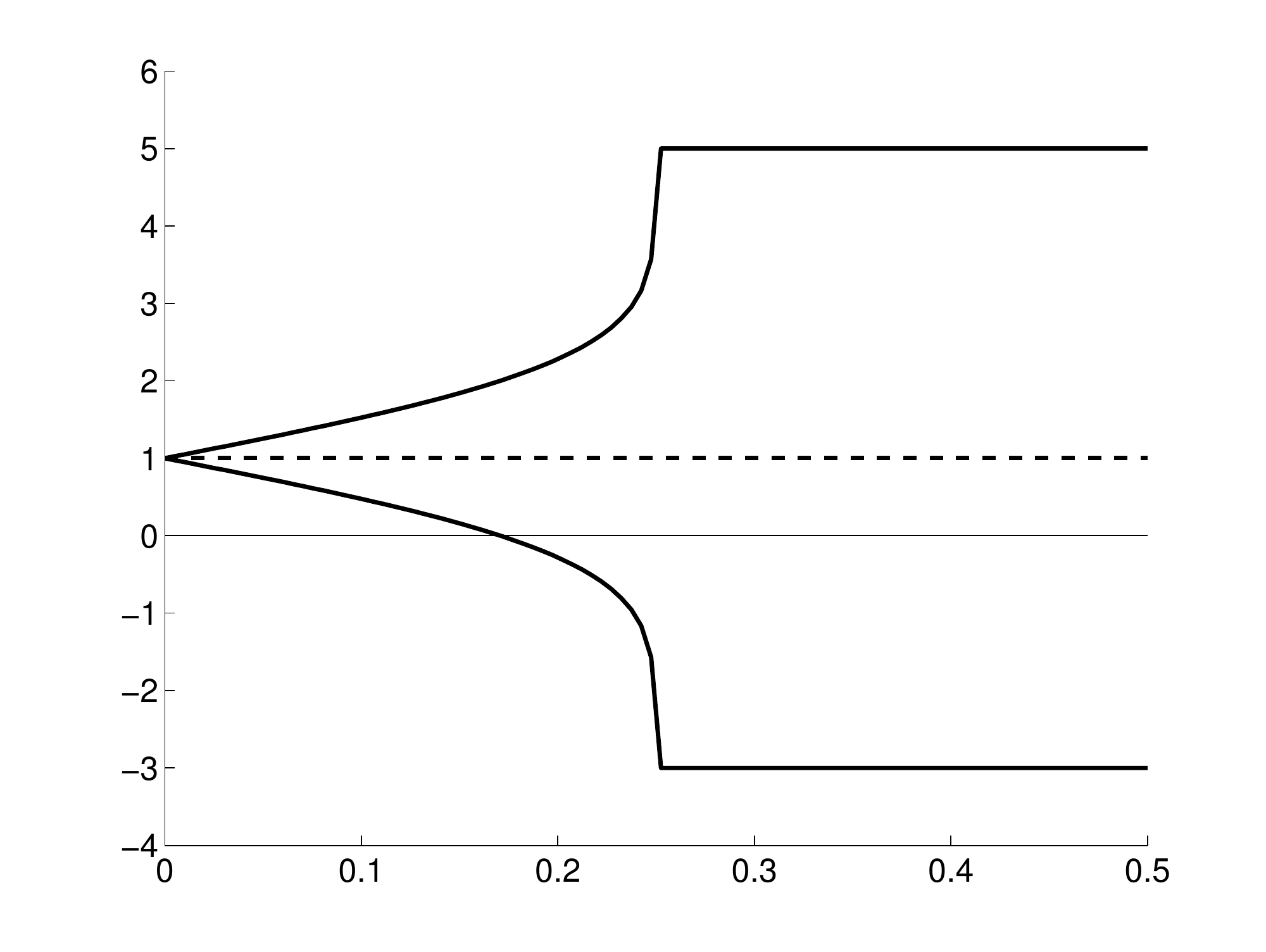}
\includegraphics[width=80mm]{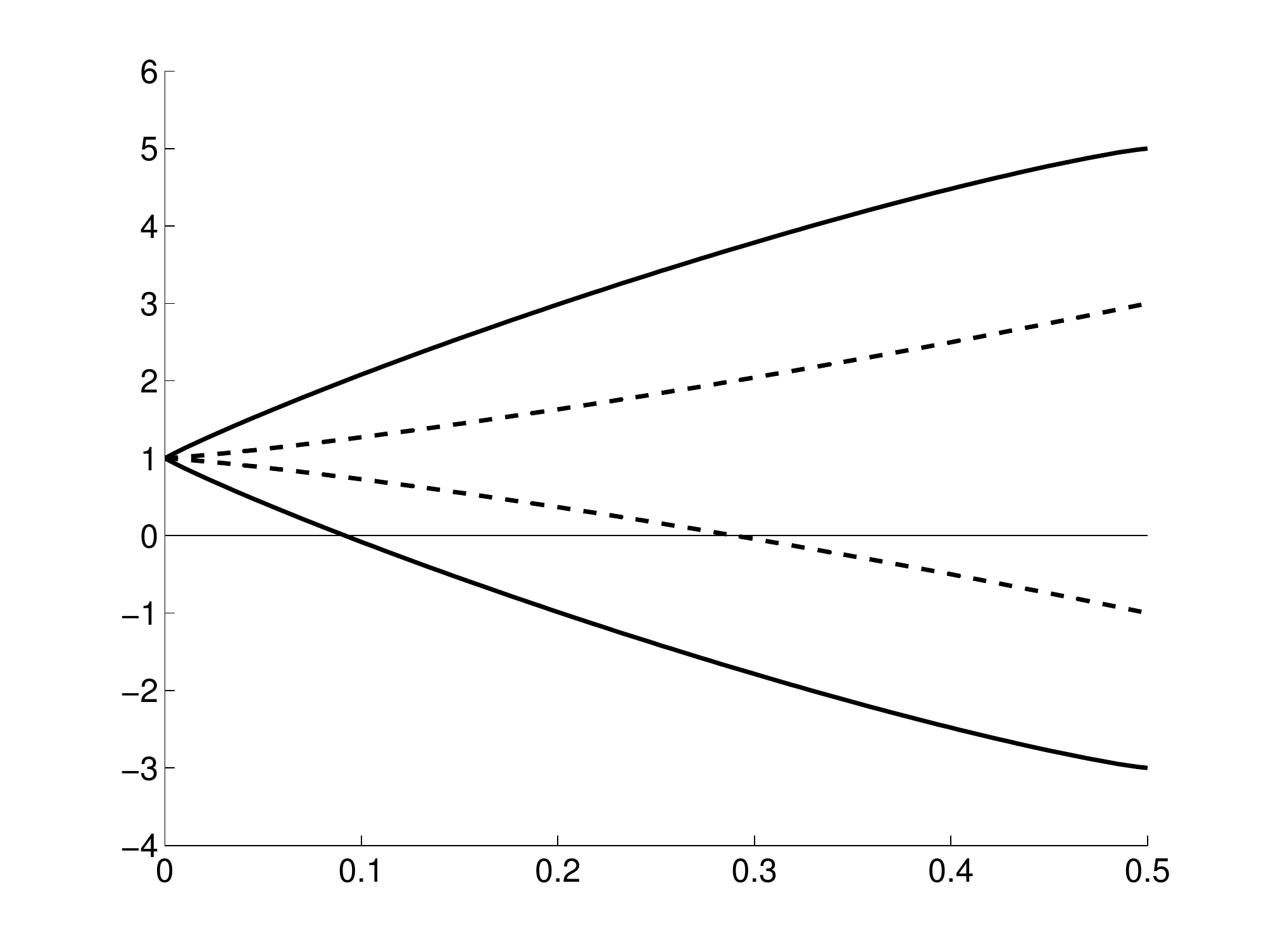}
\caption{Bounds on $\text{QTT}(0.5)$ (left) and ATT (right) for $\mathcal{U} = \mathcal{T} = [\delta,1-\delta]$ with $\delta \in [0,0.5]$. Solid: $\mathcal{U}$-independence. Dashed: $\mathcal{T}$-independence}
\label{fig:numericalIllustration}
\end{figure}

Figure \ref{fig:numericalIllustration} shows identified sets for both ATT and $\text{QTT}(0.5)$ as $\delta$ varies from $0$ to $0.5$. First consider the plot on the left, which shows the $\text{QTT}(0.5)$ bounds. The dashed lines are the identified sets under $\mathcal{T}$-independence. Since median independence of $Y_0$ from $X$ is sufficient to point identify the conditional median $Q_{Y_0 \mid X}(0.5 \mid 1)$, median independence is also sufficient to point identify the QTT at $0.5$. Hence the identified set is a singleton for all $\delta \in [0,0.5]$. Next consider the solid lines. These are the identified sets under $\mathcal{U}$-independence. When $\delta = 0.5$, $\mathcal{U}$-independence does not impose any constraints on the model, and hence we obtain the no-assumptions bounds, which are quite wide: $[-3,5]$. If we decrease $\delta$ a small amount, thus making the $\mathcal{U}$-independence constraint nontrivial, the identified set does not change. In fact, we can impose random assignment for the middle 50\% of units (i.e., $\mathcal{U} = [0.25,0.75]$, which is $\delta = 0.25$) and still we only obtain the no-assumptions bounds. Consequently, for intervals $\mathcal{T} \subseteq [0.25,0.75]$, \emph{the point identifying power of $\mathcal{T}$-independence is due solely to the constraint it imposes on the average value of the latent propensity score outside the interval} $\mathcal{T}$, rather than the constraint that random assignment holds for units in the middle of the distribution of $Y_0$.

Next consider the plot on the right of figure \ref{fig:numericalIllustration}, which shows the ATT bounds. The dashed lines are the identified sets under $\mathcal{T}$-independence. The ATT is no longer point identified under median independence, or any set $\mathcal{T} \subsetneq (0,1)$ of quantile independence conditions; that is, the ATT is partially identified for all $\delta > 0$. Nonetheless, even median independence alone has substantial identifying power: For $\delta = 0.5$, the identified set under median independence is $[-1,3]$, whereas the no-assumptions bounds are $[-3,5]$. Thus the length of the bounds has been cut in half. For $\delta > 0$, $\mathcal{U}$-independence has non-trivial identifying power, as shown in the solid lines. However, comparing the length of these bounds to the length to the $\mathcal{T}$-independence bounds, we see that imposing the average value constraint outside the interval $[\delta,1-\delta]$ again has substantial identifying power: the $\mathcal{T}$-independence bounds are anywhere from 50\% ($\delta = 0.5)$ to almost 100\% (arbitrarily small $\delta$) smaller than the $\mathcal{U}$-independence bounds. That is, the difference in lengths increases as we get closer to independence (as $\delta$ gets smaller). Thus conclusions about ATT are substantially more sensitive to small deviations from independence which do not impose the average value constraint, compared with small deviations which do impose that constraint.

\section{Conclusion}

In this paper we studied the interpretation of quantile independence of an unobserved variable $U$ from an observed variable $X$. We considered binary, discrete, and continuous $X$. We characterized sets of such quantile independence assumptions in terms of their constraints on the distribution of $X$ given $U$. This characterization shows that quantile independence requires \emph{non-monotonic} treatment selection. For example, if any quantile independence conditions on $U \mid X$ hold then for (a) binary $X$ the probability of receiving treatment given $U=u$ must be non-monotonic in $u$, while for (b) continuous $X$ the distribution of $X \mid U=u$ cannot be stochastically monotonic in $u$. Moreover, in a numerical illustration we show that the average value constraint (which imposes this non-monotonicity) has substantial identifying power, by comparing quantile independence with a weaker version we call \emph{$\mathcal{U}$-independence}.

Any class of deviations from statistical independence of $U$ and $X$ allows for certain forms of treatment selection, in the sense that the distribution of $X \mid U=u$ depends nontrivially on $u$. Thus the choice of such deviations from independence should be driven by the class of desired selection rules one wishes to allow for. For quantile independence, we characterized this class of selection rules. This class of selection rules may be of interest in some empirical applications, but not in others. Either way, researchers should justify their choice. If the class of selection rules allowed by quantile independence is deemed undesirable, several alternatives exist. These include $\mathcal{U}$-independence as defined in this paper, and $c$-dependence as defined in \cite{MastenPoirier2017}. $c$-dependence constrains the probability of receiving treatment given one's unobservables to be not too far from the overall probability of receiving treatment, and allows for monotonic selection.

In section \ref{sec:LatentSelectionModels} we considered several standard selection models. We linked their structure to the presence or absence of monotonic treatment selection, and therefore to the plausibility of quantile independence assumptions. It would be helpful to perform a more extensive analysis for additional models. In particular, in section \ref{sec:LatentSelectionModels} we only considered the relationship between a treatment variable and an unobservable. Many models instead use quantile independence to relax statistical independence between an instrument and an unobservable. Consequently, these models allow instruments to be selected on the unobservables. Our main results in sections \ref{sec:characterizationSection} and \ref{sec:ctsX} characterize the kinds of distributions of instruments given unobservables allowed by quantile independence. Detailing how these distributions relate to economic models of instrument selection would help researchers assess the plausibility of quantile independence assumptions involving instruments.

In section \ref{sec:treatmentEffects}, we studied the identifying power of the average value constraint in the standard potential outcomes model with a binary treatment. It would be helpful to perform a similar analysis for other parameters in that model, like the ATE or QTE, and also for different models altogether. In particular, quantile independence assumptions are widely used in discrete response models (following \citealt{Manski1975,Manski1985}). While our main results in sections \ref{sec:characterizationSection} and \ref{sec:ctsX} already apply to the interpretation of quantile independence in these models, an identification analysis analogous to that in section \ref{sec:treatmentEffects} would explain the importance of the average value constraint---which requires non-monotonic treatment selection---in obtaining point identification.

\singlespacing
\bibliographystyle{econometrica}
\bibliography{QI_paper}

\appendix

\section{Partial Identification of cdfs under $\mathcal{T}$- and $\mathcal{U}$-independence}\label{sec:cdfIdent}

To obtain the identified sets in section \ref{sec:treatmentEffectBounds}, we first derive sharp bounds on cdfs under generic $\mathcal{T}$- and $\mathcal{U}$-independence. We then apply these results to obtain sharp bounds on the treatment effect parameters. To this end, in this subsection we consider the relationship between a generic continuous random variable $U$ and a binary variable $X \in \{ 0,1 \}$. We derive sharp bounds on the conditional cdf of $U$ given $X$ when (1) the marginal distributions of $U$ and $X$ are known and either (2) $U$ is $\mathcal{T}$-independent of $X$ or (2$^\prime$) $U$ is $\mathcal{U}$-independent of $X$. We obtain the identified sets in section \ref{sec:treatmentEffectBounds} by applying this general result with $U = R_x$, a scaled potential outcome variable.

Let $F_{U \mid X}(u \mid x) = \Prob(U \leq u \mid X=x)$ denote the unknown conditional cdf of $U$ given $X=x$. Let $F_U(u) = \Prob(U \leq u)$ denote the known marginal cdf of $U$. Let $p_x = \Prob(X=x)$ denote the known marginal probability mass function of $X$. Let $a,b \in \R$, $a \leq b$. We define the functions $\overline{F}_{U \mid X}^\mathcal{T}(u \mid x)$, $\underline{F}_{U \mid X}^\mathcal{T}(u \mid x)$, $\overline{F}_{U \mid X}^\mathcal{U}(u \mid x)$, and $\underline{F}_{U \mid X}^\mathcal{U}(u \mid x)$ in appendix \ref{appendix:cdfbounds}. These are piecewise linear functions which depend on $a$, $b$, and $p_x$. Figure \ref{TauDepBoundsOnCDFs} plots several examples.

\begin{proposition}\label{prop:TauCDFbounds}
Suppose the following hold:
\begin{enumerate}
\item The marginal distributions of $U$ and $X$ are known.
\item $U$ is continuously distributed.
\item $p_1 \in (0,1)$.
\item $U$ is $\mathcal{T}$-independent of $X$ with $\mathcal{T} = [a,b]$.
\end{enumerate}
Let $\mathcal{F}_{\supp(U)}$ denote the set of all cdfs on $\supp(U)$. Then, for each $x \in \{0,1\}$, $F_{U \mid X}(\cdot \mid x) \in \mathcal{F}_{U \mid x}^{\mathcal{T}}$, where
\[
	\mathcal{F}_{U \mid x}^\mathcal{T} =
	\left\{ F \in \mathcal{F}_{\supp(U)} :
	\underline{F}_{U \mid X}^\mathcal{T}(u \mid x) \leq F_{U \mid X}(u \mid x) \leq \overline{F}_{U \mid X}^\mathcal{T}(u \mid x) \text{ for all $u \in \supp(U)$} \right\}.
\]
Furthermore, for each $\varepsilon \in [0,1]$, there exists a joint distribution of $(U,X)$ consistent with assumptions 1--4 above and such that
\begin{multline}\label{eq:joint epsilon cdf}
	\big( \Prob(U\leq u \mid X=1), \, \Prob(U\leq u \mid X=0) \big) \\
	= \left( \varepsilon  \underline{F}_{U \mid X}^\mathcal{T}(u \mid 1) + (1-\varepsilon) \overline{F}_{U \mid X}^{\mathcal{T}}(u \mid 1), \, 
		(1-\varepsilon)  \underline{F}_{U \mid X}^\mathcal{T}(u \mid 0) + \varepsilon \overline{F}_{U \mid X}^{\mathcal{T}}(u \mid 0) \right)
\end{multline}
for all $u \in \supp(U)$. Finally, the entire theorem continues to hold if we replace $\mathcal{T}$ with $\mathcal{U}$.
\end{proposition}

Consider the $\mathcal{T}$-independence case. Then proposition \ref{prop:TauCDFbounds} has two conclusions. First, we show that the functions $\overline{F}_{U \mid X}^\mathcal{T}(\cdot \mid x)$ and $\underline{F}_{U \mid X}^\mathcal{T}(\cdot \mid x)$ bound the unknown conditional cdf $F_{U \mid X}(\cdot \mid x)$ uniformly in their arguments. Second, we show that these bounds are functionally sharp in the sense that the joint identified set for the two conditional cdfs $(F_{U \mid X}(\cdot \mid 1), F_{U \mid X}(\cdot \mid 0))$ contains linear combinations of the bound functions $\overline{F}_{U \mid X}^\mathcal{T}(\cdot \mid x)$ and $\underline{F}_{U \mid X}^\mathcal{T}(\cdot \mid x)$. We use this second conclusion to prove sharpness of our treatment effect parameters in section \ref{sec:treatmentEffectBounds}. Identical conclusions hold in the $\mathcal{U}$-independence case.

For simplicity we have only stated this result when $\mathcal{T}$ and $\mathcal{U}$ are closed intervals $[a,b]$. It can be generalized, however. For example, for $\mathcal{T}$-independence, theorem 2 of \cite{MastenPoirier2016} provides cdf bounds when $\mathcal{T}$ is a finite union of closed intervals.

\begin{figure}[t]
\centering
\includegraphics[width=50mm]{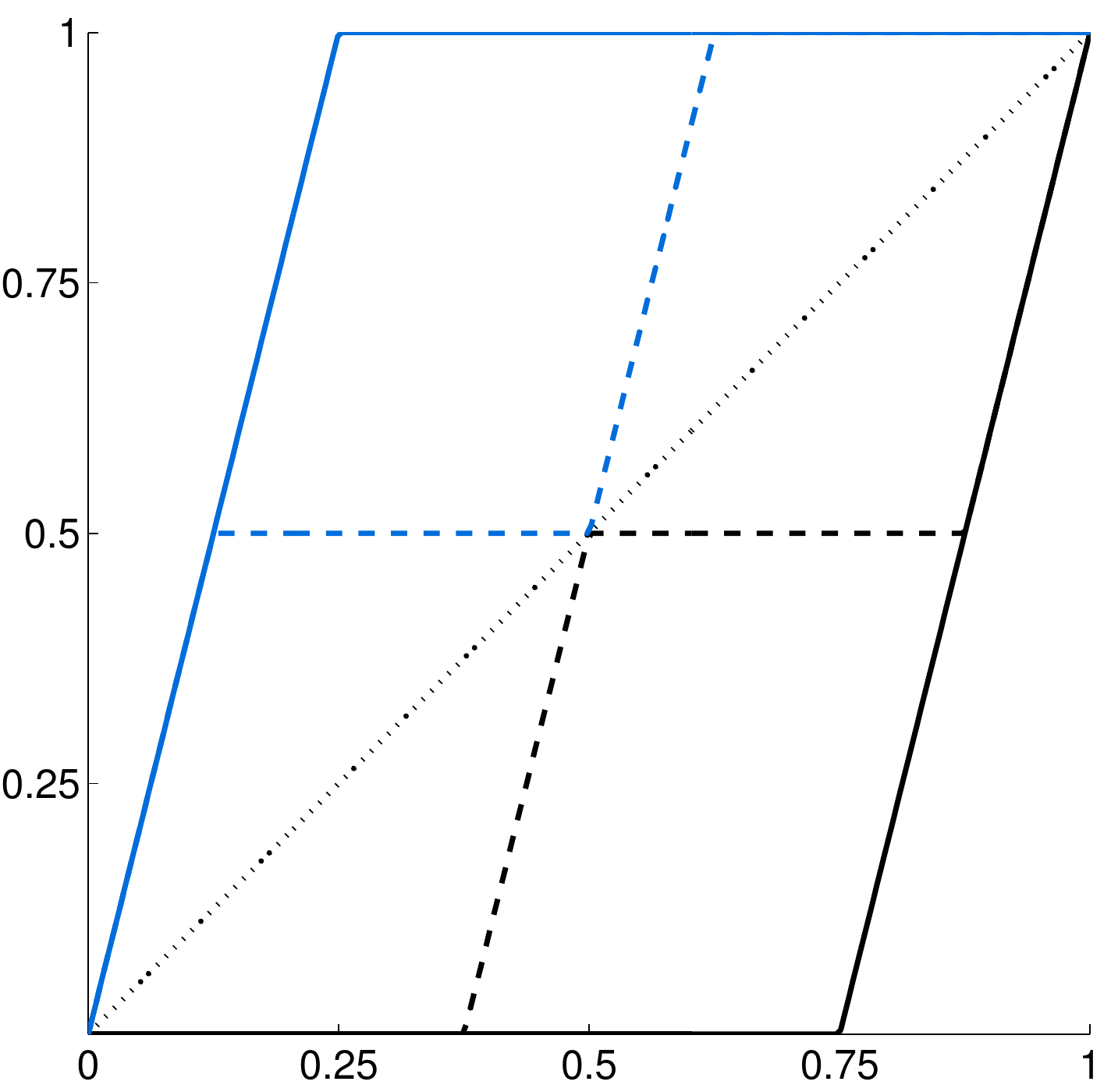}
\hspace{3mm}
\includegraphics[width=50mm]{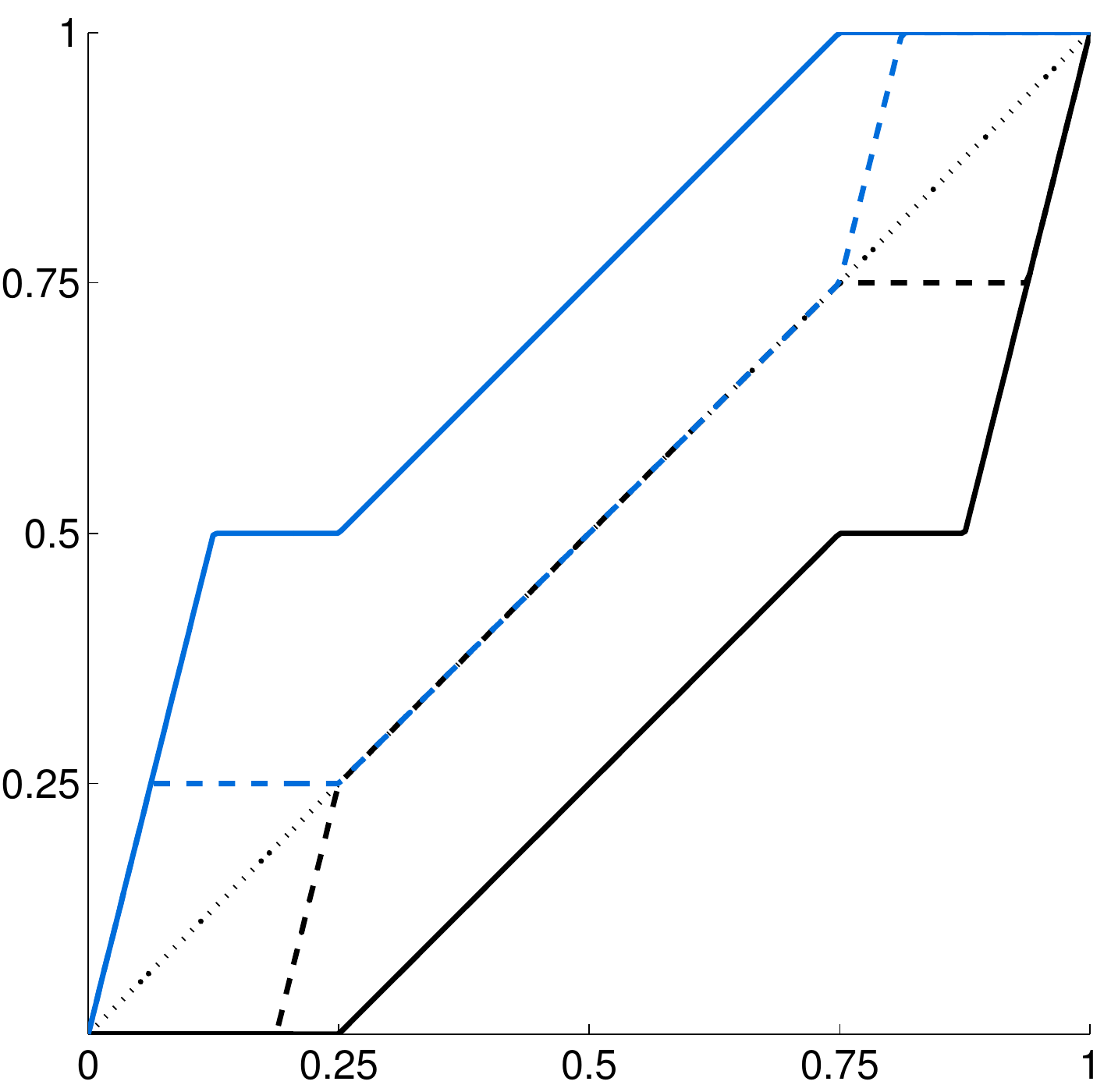}
\hspace{3mm}
\includegraphics[width=50mm]{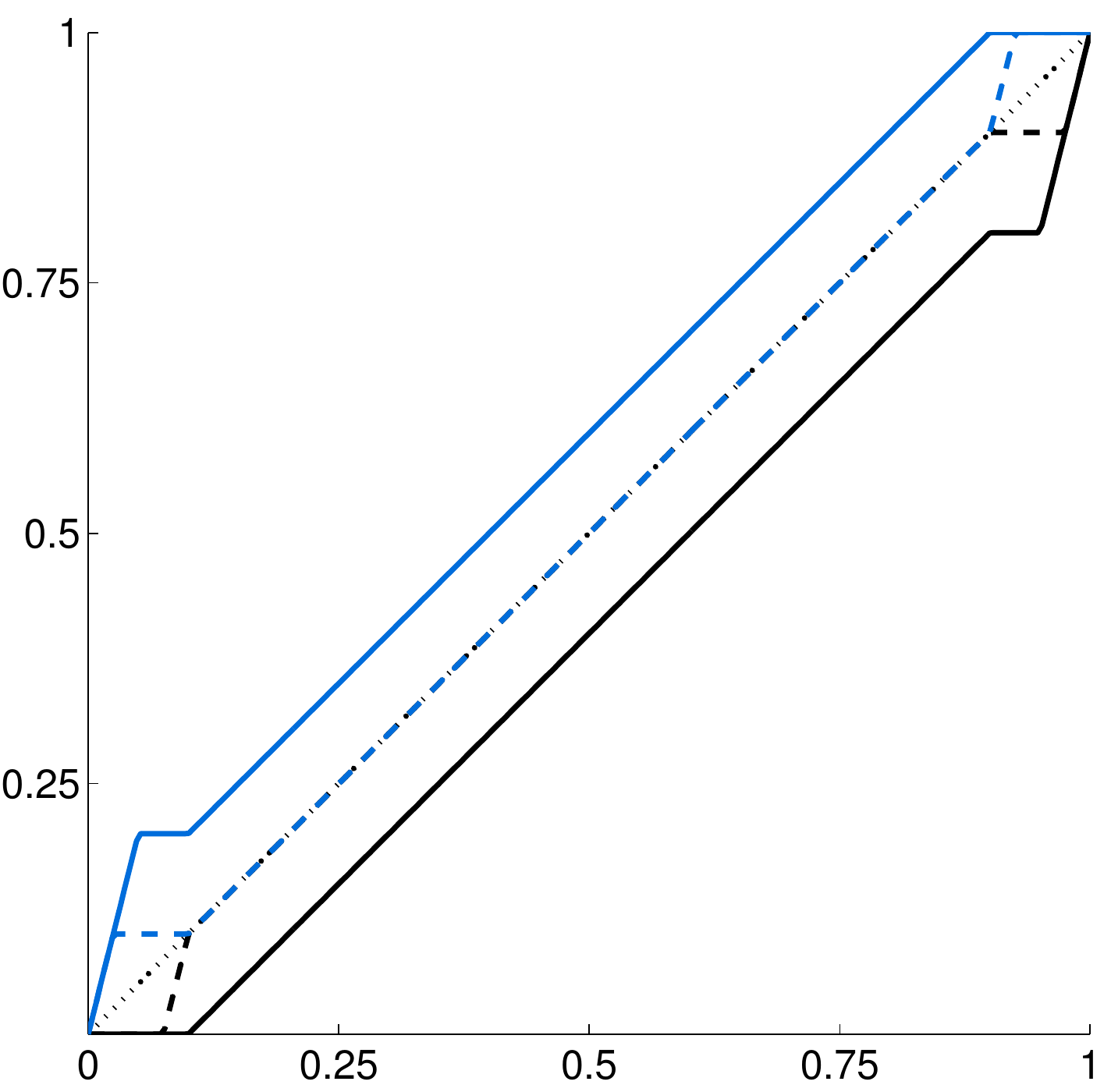}
\caption{Example upper and lower bounds on $F_{U \mid X}(u \mid 0)$, when $p_1 = 0.75$ and $U \sim \text{Unif}[0,1]$. Solid: $\mathcal{U}$-independence. Dashed: $\mathcal{T}$-independence. All three plots have $\mathcal{U} = \mathcal{T}$, for three different choices: $\{ 0.5 \}$ on the left, $[0.25, 0.75]$ in the middle, and $[0.1,0.9]$ on the right. The diagonal, representing the choice $[0,1]$---the case of full independence---is plotted as a dotted line.}
\label{TauDepBoundsOnCDFs}
\end{figure}

\section{Definitions of the bound functions}\label{appendix:cdfbounds}

In this appendix we provide the precise functional forms for the cdf bounds of proposition \ref{prop:TauCDFbounds}, the quantile bounds of proposition \ref{prop:quantileTrtBounds}, and the conditional mean bounds of corollary \ref{corr:ATT_Y0bounds}.

\subsection*{The cdf bounds}

The $\mathcal{T}$-independence bounds are as follows:
\[
	\overline{F}_{U \mid X}^\mathcal{T}(u \mid x)
		=
		\begin{cases}
			\dfrac{F_U(u)}{p_x}
				&\text{if $u \leq  Q_U(p_x F_U(a))$} \\
			F_U(a)
				&\text{if $ Q_U(p_x F_U(a)) \leq u \leq a$} \\
			F_U(u)
				& \text{if $a \leq u \leq b$} \\
			\dfrac{F_U(u) - F_U(b)}{p_x} + F_U(b)
				&\text{if $b \leq u \leq Q_{U}(p_x + F_U(b)(1-p_x))$} \\
			1
				&\text{if $Q_{U}(p_x + F_U(b)(1-p_x))  \leq u$}
		\end{cases}
\]
and
\[
	\underline{F}_{U \mid X}^\mathcal{T}(u \mid x)
		=
		\begin{cases}
			0
				&\text{if $ u \leq  Q_U((1-p_x) F_U(a))$} \\
			\dfrac{F_U(u)-F_U(a)}{p_x} + F_U(a)
				&\text{if $ Q_U((1-p_x) F_U(a)) \leq u \leq a$}\\
			F_U(u)
				& \text{if $ a \leq u \leq b$}\\
			F_U(b)
				&\text{if $b \leq u \leq Q_U(p_x F_U(b) + (1-p_x)) $} \\
			\dfrac{F_U(u)-1}{p_x} + 1
				&\text{if $Q_U(p_x F_U(b) + (1-p_x))  \leq u$}.
		\end{cases}
\]
For $\mathcal{U}$-independence, first consider the lower bound. There are two separate cases. First, if $(1-(F_U(b)-F_U(a)))(1-p_x) \leq F_U(a)$,
\begin{multline*}
	\underline{F}_{U \mid X}^\mathcal{U}(u \mid x)
	= \\
	\begin{cases}
		0 
			&\text{ for $u \leq Q_U((1-(F_U(b)-F_U(a)))(1-p_x))$} \\
		\dfrac{F_U(u) - (1-(F_U(b)-F_U(a)))(1-p_x)}{p_x}
			&\text{ for $u\in[Q_U((1-(F_U(b)-F_U(a)))(1-p_x)),a]$} \\
		\dfrac{(F_U(b)-1)(1-p_x) }{p_x} + F_U(u)
			&\text{ for $u \in [a,b]$} \\
		\dfrac{F_U(u)-1}{p_x} +1
			&\text{ for $u \geq b$}.
	\end{cases}
\end{multline*}
Second, if $(1-(F_U(b)-F_U(a)))(1-p_x) \geq F_U(a)$,
\[
	\underline{F}_{U \mid X}^\mathcal{U}(u \mid x) =
	\begin{cases}
		0
			&\text{ for $u\leq a$} \\
		F_U(u)-F_U(a)
			&\text{ for $u\in[a,b]$} \\
		F_U(b)-F_U(a)
			&\text{ for $u\in[b,Q_U(p_x(F_U(b)-F_U(a)) + 1-p_x)]$} \\
		\dfrac{F_U(u)-1}{p_x} +1
			&\text{ for $u \geq Q_U(p_x(F_U(b)-F_U(a)) + 1-p_x)$.}
	\end{cases}
\]
Next consider the upper bound. Again, there are two separate cases. First, if $(1-(F_U(b)-F_U(a)))p_x \leq F_U(a)$,
\[
	\overline{F}_{U \mid X}^\mathcal{U}(u \mid x) =
	\begin{cases}
		\dfrac{F_U(u)}{p_x}
			&\text{ for $u \leq Q_U((1-(F_U(b) - F_U(a)))p_x)$} \\
		1-(F_U(b)-F_U(a))
			&\text{ for $u \in [Q_U((1-(F_U(b) - F_U(a)))p_x),a]$} \\
		1-(F_U(b)-F_U(u))
			&\text{ for $u \in [a,b]$} \\
		1
			&\text{ for $u \geq b$.}
	\end{cases}
\]
Second, if $(1-(F_U(b)-F_U(a)))p_x \geq F_U(a)$,
\begin{multline*}
	\overline{F}_{U \mid X}^\mathcal{U}(u \mid x) = \\
	\begin{cases}
		\dfrac{F_U(u)}{p_x}
			&\text{ for $u \leq a$} \\
		\dfrac{F_U(a)}{p_x} + F_U(u)-F_U(a)
			&\text{ for $u \in [a,b]$} \\
		\dfrac{(F_U(b)-F_U(a))(p_x-1) +F_U(u)}{p_x}
			&\text{ for $u \in [b,Q_U((F_U(b)-F_U(a))(1-p_x)+ p_x)]$} \\
		1
			&\text{ for $u \geq Q_U((F_U(b)-F_U(a))(1-p_x)+ p_x)$.}
	\end{cases}
\end{multline*}

\subsection*{The quantile bounds}

The $\mathcal{T}$-independence bounds are defined by
\[
	\overline{Q}_{Y_0 \mid X}^\mathcal{T}(\tau \mid 1)
	=
	\begin{cases}
		Q_{Y \mid X}(a \mid 0)
			&\text{ for $\tau \in (0,a]$} \\
		Q_{Y \mid X}(\tau \mid 0)
			&\text{ for $\tau\in(a,b]$} \\
		Q_{Y \mid X}(1 \mid 0)
			&\text{ for $\tau\in(b,1)$},
	\end{cases}
	\qquad
	\underline{Q}_{Y_0 \mid X}^\mathcal{T}(\tau \mid 1)
	=
	\begin{cases}
		Q_{Y \mid X}(0 \mid 0)
			&\text{ for $\tau\in(0,a]$} \\
		Q_{Y \mid X}(\tau \mid 0)
			&\text{ for $\tau\in(a,b]$} \\
		Q_{Y \mid X}(b \mid 0)
			&\text{ for $\tau \in(b,1)$}.
	\end{cases}
\]
For $\mathcal{U}$-independence, there are two cases. First consider the lower bound. If $(1-(b-a))p_0 \leq a$,
\[
	\underline{Q}_{Y_0 \mid X}^\mathcal{U}(\tau \mid 1)
	=
	\begin{cases}
		Q_{Y \mid X}(0 \mid 0)
			&\text{ for $\tau \in (0,1-(b-a)]$} \\
		Q_{Y \mid X}\left( \tau + \dfrac{b-1}{p_0} \mid 0 \right)
			&\text{ for $\tau \in (1-(b-a), 1)$}.
	\end{cases}
\]
If $(1-(b-a))p_0 \geq a$,
\[
	\underline{Q}_{Y_0 \mid X}^\mathcal{U}(\tau \mid 1)
	=
	\begin{cases}
		Q_{Y \mid X}(0 \mid 0)
			&\text{ for $\tau \in \left(0,\dfrac{a}{p_1} \right]$} \\
		Q_{Y \mid X} \left(\tau - \dfrac{a}{p_1} \mid 0 \right)
			&\text{ for $\tau \in \left( \dfrac{a}{p_1},\dfrac{a}{p_1} + b-a \right]$} \\
		Q_{Y \mid X}(b-a \mid 0)
			&\text{ for $\tau \in \left( \dfrac{a}{p_1} + b-a,1 \right)$}.
	\end{cases}
\]
Next consider the upper bound. If $(1-(b-a))p_1 \leq a$,
\[
	\overline{Q}_{Y_0 \mid X}^\mathcal{U}(\tau \mid 1)
	= \begin{cases}
		 Q_{Y \mid X}(1-(b-a) \mid 0)
		 &\text{for $\tau \in \left( 0,1-(b-a) - \dfrac{1-b}{p_1} \right]$} \\
		 Q_{Y \mid X} \left( \tau +\dfrac{1-b}{p_1} \mid 0 \right)
		 &\text{for $\tau \in \left( 1-(b-a) - \dfrac{1-b}{p_1}, 1 - \dfrac{1-b}{p_1} \right]$} \\
		 Q_{Y \mid X}(1 \mid 0)
		 &\text{for $\tau \in \left( 1 - \dfrac{1-b}{p_1}, 1 \right)$.}
	\end{cases}
\]
If $(1-(b-a))p_1 \geq a$,
\[
	\overline{Q}_{Y_0 \mid X}^\mathcal{U}(\tau \mid 1)
	=
	\begin{cases}
	Q_{Y \mid X} \left(\tau + \dfrac{a}{p_0} \mid 0 \right)
		&\text{ for $\tau \in (0, b-a]$} \\
	Q_{Y \mid X}(1 \mid 0)
		&\text{ for $\tau \in (b-a,1)$}.
	\end{cases}
\]

\subsection*{The conditional mean bounds}

By integrating the quantile bounds as in the statement of corollary \ref{corr:ATT_Y0bounds}, we obtain the bounds on $\Exp(Y_0 \mid X=1)$. We provide the explicit form of these bounds but omit the derivations for brevity. For $\mathcal{T}$-independence,
\[
	\overline{\Exp}^\mathcal{T}(Y_0 \mid X=1) =
	a Q_{Y \mid X}(a \mid 0) + \int_a^b Q_{Y \mid X}(\tau \mid 0) \; d\tau + (1-b)Q_{Y \mid X}(1 \mid 0)
\]
and
\[
	\underline{\Exp}^\mathcal{T}(Y_0 \mid X=1) =
	a Q_{Y \mid X}(0 \mid 0) + \int_a^b Q_{Y \mid X}(\tau \mid 0) \; d\tau + (1-b)Q_{Y \mid X}(b \mid 0).
\]
For $\mathcal{U}$-independence, first consider the lower bound. There are two cases. If $(1-(b-a))p_0 \leq a$,
\[
	\underline{\Exp}^\mathcal{U}(Y_0 \mid X=1)
	= (1-(b-a))Q_{Y \mid X}(0 \mid 0)
	+ \int_{1-(b-a)	+ \frac{b-1}{p_0}}^{\frac{p_0 + b-1}{p_0}} Q_{Y \mid X}(\tau \mid 0) \; d\tau.
\]
If $(1-(b-a))p_0 \geq a$,
\[
	\underline{\Exp}^\mathcal{U}(Y_0 \mid X=1) = \frac{a}{p_1} Q_{Y \mid X}(0 \mid 0)
	+ \int_0^{b-a}Q_{Y \mid X}(\tau \mid 0) \; d\tau
	+ \left(1-(b-a) - \frac{a}{p_1} \right) Q_{Y \mid X}(b-a \mid 0).
\]
Next consider the upper bound. If $(1-(b-a))p_1 \leq a$,
\begin{multline*}
	\overline{\Exp}^\mathcal{U}(Y_0 \mid X=1)
	= \\
	Q_{Y \mid X}(1-(b-a) \mid 0) \left( 1-(b-a) - \frac{1-b}{p_1} \right) + \int_{1-(b-a)}^1Q_{Y \mid X}(\tau\mid 0) \; d\tau + Q_{Y \mid X}(1 \mid 0)\frac{1-b}{p_1}.
\end{multline*}
If $(1-(b-a))p_1 \geq a$,
\[
	\overline{\Exp}^\mathcal{U}(Y_0 \mid X=1)
	= \int_{\frac{a}{p_0}}^{b-a + \frac{a}{p_0}}Q_{Y \mid X}(\tau\mid 0) \; d\tau + (1-(b-a))Q_{Y \mid X}(1 \mid 0).
\]

\section{Proofs}\label{sec:proofs}

\subsection*{Proofs for section \ref{sec:characterizationSection}}

\begin{proof}[Proof of theorem \ref{thm:AvgValueCharacterization}]
This result follows immediately from our more general result for discrete $X$, theorem \ref{thm:AvgValueCharacterization_discreteX}.
\end{proof}

\begin{proof}[Proof of corollary \ref{corr:monotonicPropensityScores}]
Without loss of generality, suppose $p$ is weakly increasing. Then for any $\tau$, the average value of the propensity score to the left of $\tau$ is weakly smaller than the average value to the right:
\[
	\frac{1}{\tau} \int_0^\tau p(u) \; du \leq \frac{1}{1-\tau} \int_\tau^1 p(u) \; du.
\]
Moreover, this inequality must actually be strict for all $\tau \in (0,1)$. To see this, suppose there exists a $\tau^*\in(0,1)$ such that
\[
	\frac{1}{\tau^*} \int_0^{\tau^*} p(u) \; du 
	= \frac{1}{1-\tau^*} \int_{\tau^*}^1 p(u) \; du.
\]
This equality is equivalent to $f(\tau^*) = 0$ where we defined
\[
	f(\tau)
	=
	\int_0^{\tau} p(u)\;du - \tau \int_0^1 p(u) \; du.
\]
$f$ is differentiable with derivative
\[
	f'(\tau) = p(\tau) -\int_0^1 p(u) \; du.
\]
Moreover, $f(0) = 0$ and $f(1) = 0$. Since $p(\tau)$ is not constant on $(0,1)$, there exists a $\tau_1\in(0,1)$ small enough such that 
\[
	p(\tau_1) < \int_0^1 p(u) \; du
\]
and a $\tau_2 \in (\tau_1,1)$  large enough such that 
\[
	p(\tau_2) >  \int_0^1 p(u) \; du.
\]
Hence $f'(\tau_1) < 0$ and $f'(\tau_2) > 0$. Moreover, $f'$ is nondecreasing since $p$ is nondecreasing. Therefore $f(\tau) > 0$ for all $\tau \in(0,1)$. Hence such a $\tau^*$ cannot exist.
\end{proof}

\begin{proof}[Proof of corollary \ref{corr:SignChanges}]
For each interval $\mathcal{U}_k$, we just repeat the argument of corollary \ref{corr:monotonicPropensityScores}, conditional on $U \in \mathcal{U}_k$, noting that a nontrivial $\tau$-cdf independence condition will still hold conditional on $U \in \mathcal{U}_k$.
\end{proof}

\begin{proof}[Proof of corollary \ref{corr:TauIndep_ExtremeValues}]
Let $[a,b] \subseteq (0,1) \setminus \mathcal{T}$ with $a < b$. Consider the propensity score
\[
	p(u)
	=
	\begin{cases}
		1 &\text{if $u \in \left[a,a + \Prob(X=1) (b - a)\right)$} \\
		0 &\text{if $u \in [a + \Prob(X=1) (b - a),b]$} \\
        		\Prob(X=1) &\text{if $u \notin [a,b]$}.
	\end{cases}
\]
By definition, $p$ attains the values 0 and 1 over intervals which have positive Lebesgue measure. Next we show that $\mathcal{T}$ independence holds. Let $t_1$ and $t_2$ be any two values in $\mathcal{T}$ such that $t_1 < t_2$. Then
\[
	\frac{1}{t_2-t_1}\int_{t_1}^{t_2} p(u) \; du = \Prob(X=1)
\]
if $t_1 < t_2 < a$ or $b < t_1 < t_2$. This condition also holds if $t_1<a<b<t_2$ since
\begin{align*}
	\frac{1}{t_2-t_1}\int_{t_1}^{t_2} p(u) \; du
		&= \frac{1}{t_2-t_1} \Big( \Prob(X=1) (a-t_1) + 1 \cdot \big( \Prob(X=1)(b - a) + a - a \big) + \Prob(X=1)(t_2 - b) \Big) \\
		&= \Prob(X=1).
\end{align*}
Thus $\mathcal{T}$-independence holds by theorem \ref{thm:AvgValueCharacterization}.
\end{proof}

\begin{proof}[Proof of proposition \ref{prop:meanIndepConstraint}]
By mean independence,
\[
	\int_{-\infty}^\infty u \; f_{U \mid X}(u \mid 1) \; du = \Exp(U).
\]
By Bayes' rule,
\[
	f_{U \mid X}(u \mid 1) = \frac{p(u) f_U(u)}{\Prob(X=1)}.
\]
Substitute this expression into the integral and rearrange to obtain the result.
\end{proof}

The following lemma provides a useful alternative characterization of the constraint on the latent propensity score in proposition \ref{prop:meanIndepConstraint}.

\begin{lemma}\label{lemma:moment mean-indep}
Suppose $U$ is continuously distributed with finite mean. Suppose $X$ is binary. Then $U$ is mean independent of $X$ if and only if $\cov(p(U),U) = 0$.
\end{lemma}

\begin{proof}[Proof of lemma \ref{lemma:moment mean-indep}]
Since $X$ is binary,
\[
	\Exp[U X] = \Exp(U \mid X=1)\Prob(X=1).
\]
By the law of iterated expectations,
\[
	\Exp[U p(U)] - \Exp(U \mid X=1) \Exp[p(U)] = 0.
\]
If mean independence holds, then the left hand side is precisely $\cov(p(U),U)$. Conversely, if $\cov(p(U),U) = 0$ then
\[
	\Exp[U p(U)] = \Exp(U) \Exp(X).
\]
Since the left hand side equals $\Exp(UX)$ and since $\Exp(X) = \Prob(X=1)$,
\[
	\Exp(U \mid X=1) \Exp(X) = \Exp(U) \Exp(X).
\]
Dividing by $\Exp(X)$ shows that mean independence holds (If $\Exp(X) = 0$ then $X$ is degenerate on zero and mean independence holds trivially).
\end{proof}

\begin{proof}[Proof of corollary \ref{corr:mean-indep prop scores}]
We have
\begin{align*}
	\cov(U,p(U))
		&= \Exp[(U-\Exp(U)) p(U)] \\
		&= \Exp[(U-\Exp(U))(p(U) - p(\Exp(U)))].
\end{align*}
Without loss of generality, suppose $p(u)$ is non-decreasing on $\supp(U)$. Therefore
\[
	(U-\Exp(U))(p(U) - p(\Exp(U))) \geq 0
\]
holds with probability one. Moreover, equality holds with probability equal to $\Prob[p(U) = p(\Exp(U))]$. Since $p$ is non-constant and non-decreasing, the probability that $p(U)$ is equal to a constant is strictly less than one. Hence $\cov(U,p(U)) > 0$. Therefore, by lemma \ref{lemma:moment mean-indep}, $U$ cannot be mean independent of $X$.
\end{proof}

\subsection*{Proofs for section \ref{sec:ctsX}}

\begin{proof}[Proof of theorem \ref{thm:AvgValueCharacterization_ctsX}]
If $x$ is such that $\Prob(X > x) \in \{ 0, 1 \}$ then equation \ref{eq:averageValueCondition_main_ctsX} holds trivially. So suppose $x$ is such that $\Prob(X > x) \in (0,1)$. Define $\widetilde{X} = \indicator(X > x)$. Note that $\Prob(\widetilde{X}=1) \in (0,1)$ and $\Prob(\widetilde{X}=0) \in (0,1)$. The result now follows from applying theorem \ref{thm:AvgValueCharacterization_discreteX} with $\widetilde{X}$.
\end{proof}

\begin{proof}[Proof of corollary \ref{corr:noRegDependence}]
Follows by defining $\widetilde{X} = \indicator(X > x)$ and applying corollary \ref{corr:monotonicPropensityScores}.
\end{proof}

The following lemma provides an alternative way of writing the average value constraint \eqref{eq:averageValueCondition_main}. We use this in the proof of theorem \ref{thm:AvgValueCharacterization_discreteX}.

\begin{lemma}\label{lemma:subgroupPropensityScore}
Suppose $X$ is discrete with support $\{x_1,\ldots,x_K\}$ and $\Prob(X=x_k) \in(0,1)$ for $k \in \{1,\ldots,K \}$. Suppose $U$ is continuously distributed; normalize $U \sim \text{Unif}[0,1]$. For $t_1, t_2 \in [0,1]$ with $t_1< t_2$ we have
\[
	\Prob(X = x_k \mid U \in [t_1,t_2]) = \frac{1}{t_2 - t_1} \int_{[t_1,t_2]} \Prob(X=x_k \mid U=u) \; du.
\]
\end{lemma}

\begin{proof}[Proof of lemma \ref{lemma:subgroupPropensityScore}]
We have
\begin{align*}
	\Prob(X = x_k \mid U \in [t_1,t_2])
		&= \frac{\Prob(X=x_k,t_1 \leq U \leq t_2)}{\Prob(t_1 \leq U \leq t_2)} \\
		&= \frac{1}{t_2 - t_1} \Prob(X=x_k, t_1 \leq U \leq t_2) \\
		&= \frac{1}{t_2 - t_1} \int_{t_1}^{t_2} f_{X,U}(x_k,u) \; du \\
		&= \frac{1}{t_2 - t_1} \int_{t_1}^{t_2} \frac{f_{X,U}(x_k,u)}{f_U(u)} \; du \\
		&= \frac{1}{t_2 - t_1} \int_{t_1}^{t_2} \Prob(X=x_k \mid U=u) \; du.
\end{align*}
The fourth line follows since $f_U(u) = 1$ by $U \sim \text{Unif}[0,1]$.
\end{proof}

\begin{lemma}\label{lemma:continuity}
Suppose $U$ is continuously distributed. Suppose $X$ is discrete with support $\{ x_1,\ldots,x_K \}$ and probability masses $p_k = \Prob(X=x_k) \in (0,1)$ for $k \in \{ 1,\ldots, K \}$. Then $F_{U \mid X}(\cdot \mid x_k)$ is a continuous function for all $k \in \{ 1,\ldots,K \}$.
\end{lemma}

\begin{proof}[Proof of lemma \ref{lemma:continuity}]
Suppose by way of contradiction that $F_{U \mid X}(\cdot \mid x_k)$ is not continuous at some point $u^*$. Since cdfs are right-continuous, we must have
\[
	\lim_{ u \nearrow u^*} F_{U \mid X}(u \mid x_k) < F_{U \mid X}(u^* \mid x_k).
\]
This implies
\[
	\Prob(U = u^* \mid X = x_k) > 0.
\]
Therefore
\begin{align*}
	0
		&= \Prob(U = u^*)
			&\text{by $U$ continuously distributed} \\
		&= \sum_{j=1}^K \Prob( U = u^* \mid X = x_j) p_j
			&\text{by the law of total probability} \\
		&\geq \Prob( U = u^* \mid X = x_k) p_k \\
		&> 0.
\end{align*}
This is a contradiction.
\end{proof}

\begin{proof}[Proof of theorem \ref{thm:AvgValueCharacterization_discreteX}]
\hfill
\begin{itemize}
\item[($\Rightarrow$)] Suppose $U$ is $\mathcal{T}$-independent of $X$. Let $t_1, t_2 \in\mathcal{T} \cup \{ 0, 1 \}$ with $t_1 < t_2$. Then, for any $x \in \supp(X)$,
\begin{align*}
	\Prob(X=x \mid U \in [t_1,t_2])
		&= \frac{\Prob(X=x, U \in [t_1,t_2]) }{\Prob(U \in [t_1,t_2]) } \\
		&= \frac{\Prob(U \in [t_1,t_2] \mid X=x) \Prob(X=x)}{t_2-t_1} \\
		&= \frac{(\Prob(U \leq t_2 \mid X=x) - \Prob(U < t_1 \mid X=x)) \Prob(X=x)}{t_2-t_1} \\
    		&= \frac{(\Prob(U \leq t_2 \mid X=x) - \Prob(U \leq t_1 \mid X=x)) \Prob(X=x)}{t_2-t_1} \\
    		&= \frac{(t_2 - t_1) \Prob(X=x)}{t_2-t_1} \\
    		&= \Prob(X=x).
\end{align*}
The second line follows since $U \sim \text{Unif}[0,1]$. The fourth line follows since $U \mid X$ is continuously distributed, which itself follows by $X$ being discretely distributed and lemma \ref{lemma:continuity}. The fifth line follows from $\mathcal{T}$-independence.

\item[($\Leftarrow$)] Suppose that for any $x \in \supp(X)$,
\[
	\Prob(X=x \mid U \in [t_1,t_2]) = \Prob(X=x)
\]
for all $t_1, t_2 \in \mathcal{T} \cup \{ 0, 1 \}$ with $t_1 < t_2$. Then,
\begin{align*}
	\Prob(U \in [t_1,t_2] \mid X=x)
		&= \frac{\Prob(X=x \mid U \in[t_1,t_2]) \Prob(U \in [t_1,t_2])}{\Prob(X=x)} \\
		&= \frac{\Prob(X=x) \Prob(U \in [t_1,t_2])}{\Prob(X=x)} \\
    		&= \Prob(U \in [t_1,t_2]).
\end{align*}
The second line follows by assumption. Setting $t_1 = 0$ and using $U \sim \text{Unif}[0,1]$ gives the result.
\end{itemize}
The result now follows by lemma \ref{lemma:subgroupPropensityScore}.
\end{proof}

\begin{proof}[Proof of corollary \ref{corr:discreteXnotMonotone}]
Follows by defining $\widetilde{X} = \indicator(X\geq x_k)$ and applying corollary \ref{corr:monotonicPropensityScores}.
\end{proof}

\subsection*{Proofs for section \ref{sec:LatentSelectionModels}}

\begin{proof}[Proof of proposition \ref{prop:ImbensNeweySelection}]
Define
\[
	M(x,z,v) = \Exp[ m(x,U) \mid V=v ] - c(x,z),
\]
which equals our objective function since $Z \independent (U,V)$. We have
\begin{align*}
	\frac{\partial^2}{\partial x^2} M(x,z,v)
		&= \Exp \left[ \frac{\partial^2}{\partial x^2} m(x,U) \mid V=v \right] - \frac{\partial^2 c(x,z)}{\partial x^2} \\
		&< 0
\end{align*}
for all $z$, $x$, and $v$. The first line follows by the dominated convergence theorem. The second line follows by our assumptions that $\partial^2 c(x,z) / \partial x^2 > 0$ and $\partial^2 m(x,u)/ \partial x^2 < 0$ for all $z$, $x$, and $u$. Thus the function $M(\cdot,z,v)$ is globally strictly concave, for each $z$ and $v$.

Our assumptions on the limits of $m_1(x,u) \equiv \partial m(x,u) / \partial x$ as $x \rightarrow 0$ or $\infty$ combined with our dominance assumption imply that
\[
	\lim_{x \rightarrow \infty} \Exp[m_1(x,U) \mid V=v] \rightarrow 0
	\qquad \text{and} \qquad
	\lim_{x \rightarrow 0} \Exp[m_1(x,U) \mid V=v] \rightarrow \infty.
\]
This result combined with our assumptions on the limits of $\partial c(x,z) / \partial x$ as $x \rightarrow 0$ or $\infty$ imply that
\[
	\lim_{x \rightarrow \infty} \frac{\partial}{\partial x} M(x,z,v) = -\infty
	\qquad \text{and} \qquad
	\lim_{x \rightarrow 0} \frac{\partial}{\partial x} M(x,z,v) = \infty.
\]
Consequently, since $\partial M(x,z,v) / \partial x$ is continuous, the intermediate value theorem implies there exists a solution to the first order condition $\partial M(x,z,v) / \partial x = 0$. Since $M(\cdot,z,v)$ is globally strictly concave, this solution is unique. Let $h(z,v)$ denote this solution.

Next we show that $h(z,v)$ is strictly increasing in $v$. Let
\begin{align*}
	M_1(x,z,v)
		&= \frac{\partial}{\partial x} M(x,z,v) \\
		&= \Exp \left[ \frac{\partial}{\partial x} m(x,U) \mid V=v \right] - \frac{\partial c(x,z)}{\partial x} \\
		&= \Exp [m_1(x,U) \mid V=v] - \frac{\partial c(x,z)}{\partial x}.
\end{align*}
The second line follows by the dominated convergence theorem. Since $U$ and $V$ satisfy the strict MLRP, and since $m_1(x,u)$ is strictly increasing in $u$ for each $x$, $\Exp[m_1(x,U) \mid V=v]$ is strictly increasing in $v$; this follows by a straightforward generalization of theorem 5 on page 1100 of \cite{MilgromWeber1982}. Thus $M_1(x,z,v)$ is strictly increasing in $v$, for all $x$ and $z$. Finally, note that the optimum $h(z,v)$ is in the interior of the constraint set (which is simply $\{ x \in \R: x \geq 0 \}$). Thus $h(z,\cdot)$ is strictly increasing by theorem 1 on page 205 of \cite{EdlinShannon1998}.
\end{proof}

\begin{proof}[Proof of proposition \ref{corr:ImbensNeweySelection}]
By assumption $h(z,\cdot)$ is strictly monotone; normalize it to be strictly increasing. Without loss of generality, normalize $V \sim \text{Unif}[0,1]$. We have
\begin{align*}
	\Prob(X \leq x \mid U=u)
		&= \Prob( h(Z,V) \leq x \mid U=u) \\
		&= \Prob(V \leq h^{-1}(Z,x) \mid U=u) \\
		&= \Prob(V \leq F_{X \mid Z}(x \mid Z) \mid U=u) \\
		&= \int_{\supp(Z \mid U=u)} \Prob(V \leq F_{X \mid Z}(x \mid z) \mid U=u,Z=z) \; dF_{Z \mid U}(z \mid u) \\
		&= \int_{\supp(Z)} \Prob(V \leq F_{X \mid Z}(x \mid z) \mid U=u) \; dF_Z(z)
\end{align*}
The second line follows since $h(z,\cdot)$ is strictly increasing. The third line follows since $Q_{X \mid Z}(v \mid z) = h(z,v)$, by $Z \independent V$, $h(z,\cdot)$ strictly increasing, $V \sim \text{Unif}[0,1]$, and quantile equivariance. The fourth line follows by iterated expectations. The fifth line follows by $Z \independent (U,V)$. Suppose $V$ is positive regression dependent on $U$; the proof for the negative regression dependence case is symmetric. Then $\Prob(V \leq v \mid U=u)$ is nonincreasing in $u$ for each $v$. In particular, $\Prob(V \leq F_{X \mid Z}(x \mid z) \mid U=u)$ is nonincreasing in $u$ for all $x$ and $z$. Since the integrand is monotonic in $u$ for each $z$, the integral over $z$ is also monotonic in $u$.
\end{proof}

\begin{proof}[Proof of corollary \ref{corr:ImbensNeweySelectionFinalResult}]
Follows immediately from corollary \ref{corr:noRegDependence} and propositions \ref{prop:ImbensNeweySelection} and  \ref{corr:ImbensNeweySelection}. Note that $U$ and $V$ are not independent since they satisfy the \emph{strict} MLRP. This implies that $X$ and $U$ are not independent.
\end{proof}

\subsection*{Proofs for section \ref{sec:treatmentEffects}}

\begin{proof}[Proof of corollary \ref{prop:TauImpliesU}]
If $a = b$ the result holds trivially since $\mathcal{T}$ has measure zero. So suppose $a<b$. Then
\[
	\Prob(U\leq u \mid X=x,U\in[a,b])
	= \frac{\Prob(U \in  [0,u] \cap [a,b] \mid X=x)}{\Prob(U\in [a,b] \mid X=x)}.
\]
If $u > b$ this fraction equals 1. If $u < a$ the numerator is zero. In either case, $\Prob(U\leq u \mid X=x,U\in[a,b])$ does not depend on $x$. If $u\in[a,b]$, $[0,u] \cap [a,b] = [a,u]$. Hence
\begin{align*}
	\Prob(U\leq u \mid X=x,U\in[a,b])
		&= \frac{F_{U \mid X}(u \mid x) - F_{U \mid X}(a \mid x)}{F_{U \mid X}(b \mid x) - F_{U \mid X}(a \mid x)}\\
		&= \frac{u - a}{b - a}\\
		&= \Prob(U\leq u \mid U\in[a,b]).
\end{align*}
The first line follows since $U \mid X=x$ is continuously distributed, by lemma \ref{lemma:continuity}. The second line follows from $a,u,b\in\mathcal{T}$. Therefore, $U\independent X \mid \{ U \in \mathcal{T} \}$. This implies that, for almost all $u \in \mathcal{T}$,
\begin{align*}
	\Prob(X=1 \mid U=u)
		&= \Prob(X=1 \mid U=u,U\in\mathcal{T}) \\
		&= \Prob(X=1 \mid U\in\mathcal{T}).
\end{align*}
The second equality follows from conditional independence. Thus we have shown that $p(u)$ is flat on $\mathcal{T}$. Finally, let $t_1, t_2 \in [a,b]$ with $t_1 < t_2$. Then
\begin{align*}
	\Prob(X=1)
		&= \frac{1}{t_2 - t_1} \int_{t_1}^{t_2} p(u) \; du \\
		&= \Prob(X=1 \mid U \in \mathcal{T}).
\end{align*}
The first line follows by theorem \ref{thm:AvgValueCharacterization} while the second line follows by our derivations above showing that $p(u)$ is flat on $\mathcal{T}$. Hence
\[
	\Prob(X=1 \mid U=u) = \Prob(X=1)
\]
for almost all $u \in \mathcal{T}$.
\end{proof}

\begin{proof}[Proof of corollary \ref{prop:TauImpliesU_ctsX}]
Follows by defining $\widetilde{X} = \indicator(X\leq x)$ and applying corollary \ref{prop:TauImpliesU}.
\end{proof}

\begin{proof}[Proof of proposition \ref{prop:quantileTrtBounds}]
By the law of total probability and equation \eqref{eq:potential outcomes},
\begin{align*}
	F_{Y_0}(y)
		&= F_{Y_0 \mid X}(y \mid 1) p_1 + F_{Y_0 \mid X}(y \mid 0) p_0 \\
		&= F_{Y_0 \mid X}(y \mid 1) p_1 + F_{Y \mid X}(y \mid 0) p_0.
\end{align*}
Rearranging yields
\begin{equation}\label{eq:cdfY0givenX}
	F_{Y_0 \mid X}(y \mid 1) = \frac{F_{Y_0}(y) - p_0 F_{Y \mid X}(y \mid 0)}{p_1}.
\end{equation}
Finally, the desired quantile is simply the left-inverse of this conditional cdf:
\[
	Q_{Y_0 \mid X}(\tau \mid 1) = F_{Y_0 \mid X}^{-1}(\tau \mid 1).
\]
Thus it suffices to obtain bounds on the unconditional cdf $F_{Y_0}(y)$. By equation (12) on page 337 of \cite{MastenPoirier2017},
\begin{equation}\label{eq:eq12MP2018}
	F_{Y \mid X}(y \mid 0) = F_{R_0 \mid X}(F_{Y_0}(y) \mid 0)
\end{equation}
where $R_0 \equiv F_{Y_0}(Y_0)$ is the \emph{rank} of $Y_0$. By A\ref{assn:continuity}, $Y_0$ is continuously distributed and hence $R_0 \sim \text{Unif}[0,1]$. The main idea of the proof is that proposition \ref{prop:TauCDFbounds} yields bounds on $F_{R_0 \mid X}$, which we then invert to obtain bounds on $F_{Y_0}$. We then substitute these bounds into equation \eqref{eq:cdfY0givenX} to obtain bounds on $F_{Y_0 \mid X}(\cdot \mid 1)$. Inverting those bounds yields the quantile bounds given in appendix \ref{appendix:cdfbounds}. Since the bounds of proposition \ref{prop:TauCDFbounds} are not always uniquely invertible, we approximate them by invertible bound functions. Here we explain the main argument, but we omit the full details since these are similar to the proof of proposition 2 in \cite{MastenPoirier2017}.

\bigskip

\textbf{Under $\mathcal{T}$-independence} of $Y_0$ from $X$ with $\mathcal{T} = [Q_{Y_0}(a),Q_{Y_0}(b)]$, we have $\mathcal{T}$-independence of $R_0$ from $X$ with $\mathcal{T} = [a,b]$. To see this, let $\tau \in [a,b]$. Then
\begin{align*}
	F_{R_0 \mid X}(\tau \mid x)
		&= \Prob(R_0 \leq \tau \mid X=x) \\
		&= \Prob(F_{Y_0}(Y_0) \leq \tau  \mid X=x) \\
		&= \Prob(Y_0 \leq Q_{Y_0}(\tau) \mid X=x) \\
		&= F_{Y_0 \mid X}(Q_{Y_0}(\tau) \mid x) \\
		&= F_{Y_0}(Q_{Y_0}(\tau)) \\
		&= \tau.
\end{align*}
The second line follows by definition of the rank $R_0$. The fifth line follows by $\mathcal{T}$-independence of $Y_0$ from $X$ and since $\tau \in [a,b]$. The third and sixth lines follow by A\ref{assn:continuity}.\ref{A1_1}. Thus proposition \ref{prop:TauCDFbounds} yields sharp bounds on $F_{R_0 \mid X}$. Substituting these bounds into our argument above yields the bounds on $Q_{Y \mid X}(\tau \mid 1)$ in appendix \ref{appendix:cdfbounds}.

Sharpness of these bounds holds in the same sense as sharpness of the CQTE bounds in proposition 3 of \cite{MastenPoirier2017}. That is, the conditional quantile of $Y_0 \mid X=1$ should be a continuous and strictly increasing function, while $\underline{Q}_{Y_0 \mid X}^\mathcal{T}(\tau \mid 1)$ and $\overline{Q}_{Y_0 \mid X}^\mathcal{T}(\tau \mid 1)$ may have discontinuities  and flat regions. Nevertheless we show there exists a function that is arbitrarily close (pointwise in $\tau$) to these bounds that is continuous and strictly increasing. To see this, for $\eta \in[0,1]$ define
\begin{align*}
	\underline{F}^{\mathcal{T}}_{R_0 \mid X}(u \mid 0;\eta)
		&= (1-\eta) \cdot \underline{F}^{\mathcal{T}}_{R_0 \mid X}(u \mid 0) + \eta \cdot u \\
	\overline{F}^{\mathcal{T}}_{R_0 \mid X}(u \mid 0;\eta)
		&= (1-\eta) \cdot \overline{F}^{\mathcal{T}}_{R_0 \mid X}(u \mid 0) + \eta \cdot u.
\end{align*}
These cdfs satisfy $\mathcal{T}$-independence. For each $\eta  > 0$, they are continuous and strictly increasing. Finally, they converge uniformly to $\underline{F}^{\mathcal{T}}_{R_0 \mid X}(u \mid 0)$ and $\overline{F}^{\mathcal{T}}_{R_0 \mid X}(u \mid 0)$, respectively, as $\eta \searrow 0$. Therefore, we can substitute the cdf bounds $\underline{F}^{\mathcal{T}}_{R_0 \mid X}(u \mid 0;\eta)$ and $\overline{F}^{\mathcal{T}}_{R_0 \mid X}(u \mid 0;\eta)$ into equation \eqref{eq:eq12MP2018}, invert and then substitute that into equation \eqref{eq:cdfY0givenX} to obtain
\[
	\overline{F}^{\mathcal{T}}_{Y_0 \mid X}(y \mid 1;\eta)
	\equiv
	\frac{\underline{F}_{R_0 \mid X}^{^{\mathcal{T}}\, -1}(F_{Y \mid X}(y \mid 0) \mid 0;\eta) - p_0 F_{Y \mid X}(y \mid 0)}{p_1}
\]
and
\[
	\underline{F}^{\mathcal{T}}_{Y_0 \mid X}(y \mid 1;\eta)
	\equiv
	\frac{\overline{F}_{R_0 \mid X}^{^{\mathcal{T}}\, -1}(F_{Y \mid X}(y \mid 0) \mid 0;\eta) - p_0 F_{Y \mid X}(y \mid 0)}{p_1}.
\]
Taking the inverses of these two cdfs and letting $\eta \searrow 0$ allows us to attain points arbitrarily close to the endpoints of the set $[\underline{Q}_{Y_0 \mid X}^\mathcal{T}(\tau \mid 1), \overline{Q}_{Y_0 \mid X}^\mathcal{T}(\tau \mid 1)]$. The rest of the interior is attained by picking sufficiently small $\eta > 0$ and taking convex combinations of the bound functions, as in equation \eqref{eq:joint epsilon cdf}, and letting $\varepsilon$ vary from 0 to 1.

\bigskip

\textbf{For $\mathcal{U}$-independence} we also obtain sharpness of the interior because the functions $\underline{Q}_{Y_0 \mid X}^\mathcal{U}(\tau \mid 1)$ and $\overline{Q}_{Y_0 \mid X}^\mathcal{U}(\tau \mid 1)$ are not necessarily continuous or strictly increasing.  Nevertheless, as for $\mathcal{T}$-independence, we can obtain continuous and strictly increasing functions that are arbitrarily close (pointwise in $\tau$) to these bounds. To see this, for $\eta = (\eta_1,\eta_2)\in(0,\min\{p_1,p_0\})^2$ define
\begin{align*}
	\underline{p}_x(u;\eta)
		&= \min\{\max\{\underline{p}_x(u),\eta_1\},1-\eta_2\} \\
	\overline{p}_x(u;\eta)
		&= \min\{\max\{\overline{p}_x(u),\eta_1\},1-\eta_2\}
\end{align*}
where $\underline{p}_x$ and $\overline{p}_x$ are defined as in the proof of proposition \ref{prop:TauCDFbounds} and we let $U \equiv R_0$. 
These conditional probabilities lie in $(0,1)$ and satisfy $\mathcal{U}$-independence. Moreover, there exists an $\widetilde{\eta}_1 \in (0,\min \{p_1,p_0 \})$ such that for any $\eta_1 \in (0,\widetilde{\eta}_1)$, there is an $\eta _2(\eta_1) \in (0,\min\{p_1,p_0\})$ such that
\[
	\int_{0}^1 \underline{p}_x(u;\eta) \; du = p_x
\]
for $\eta = (\eta_1,\eta_2(\eta_1))$. This follows by the intermediate value theorem. An analogous result holds for the conditional probability $\overline{p}_x(u;\eta)$. For such values of $\eta$, define 
\[
	\underline{F}^{\mathcal{U}}_{R_0 \mid X}(u \mid 0;\eta)
	= \int_{0}^u \frac{\underline{p}_x(v;\eta)}{p_x} \; dv
	\qquad \text{and} \qquad
	\overline{F}^{\mathcal{U}}_{R_0 \mid X}(u \mid 0;\eta)
	= \int_{0}^u \frac{\underline{p}_x(v;\eta)}{p_x} \; dv.
\]
These cdf bounds are strictly increasing and continuous since the integrand is strictly positive. Therefore, we can substitute these cdf bounds into equation \eqref{eq:eq12MP2018}, invert and then substitute that into equation \eqref{eq:cdfY0givenX} to obtain
\[
	\overline{F}^{\mathcal{U}}_{Y_0 \mid X}(y \mid 1;\eta)
		\equiv \frac{
			\underline{F}_{R_0 \mid X}^{^{\mathcal{U}}\,-1}(F_{Y \mid X}(y \mid 0) \mid 0;\eta) 
				- p_0 F_{Y \mid X}(y \mid 0)
			}{
			p_1
			}
\]
and
\[
	\underline{F}^{\mathcal{U}}_{Y_0 \mid X}(y \mid 1;\eta)
		\equiv \frac{
			\overline{F}_{R_0 \mid X}^{^{\mathcal{U}}\, -1}(F_{Y \mid X}(y \mid 0) \mid 0;\eta) 
				- p_0 F_{Y \mid X}(y \mid 0)
			}{
			p_1
			}.
\]
Taking the inverses of these two cdfs and letting $\eta_1 \searrow 0$ allows us to attain points arbitrarily close to the endpoints of the set $[\underline{Q}_{Y_0 \mid X}^\mathcal{U}(\tau \mid 1), \overline{Q}_{Y_0 \mid X}^\mathcal{U}(\tau \mid 1)]$. The rest of the interior is attained by picking sufficiently small $\eta_1 > 0$ and taking convex combinations of the bound functions, as in equation \eqref{eq:joint epsilon cdf}, and letting $\varepsilon$ vary from 0 to 1.
\end{proof}

\begin{proof}[Proof of corollary \ref{corr:ATT_Y0bounds}]
This result follows by derivations similar to the proof of corollary 1 in \cite{MastenPoirier2017}.
\end{proof}

\subsection*{Proofs for appendix \ref{sec:cdfIdent}}

We frequently use the following result.

\begin{lemma}\label{lemma:cdfAsIntegralOfPropensityScore}
Let $U$ be a continuous random variable. Let $X$ be a random variable with $p_x = \Prob(X=x) > 0$. Then
\[
	F_{U \mid X}(u \mid x) = \int_{-\infty}^u \frac{\Prob(X=x \mid U=v)}{p_x} \; dF_U(v).
\]
\end{lemma}

\begin{proof}[Proof of lemma \ref{lemma:cdfAsIntegralOfPropensityScore}]
See lemma 1 in \cite{MastenPoirier2017}.
\end{proof}

\begin{proof}[Proof of proposition \ref{prop:TauCDFbounds} ($\mathcal{T}$-independence)]
We prove this statement for $\mathcal{T}$-independence first, then for $\mathcal{U}$-independence. Both proofs proceed by first deriving the upper cdf bound, then deriving the lower cdf bound, and finishing by showing the joint attainability of the cdfs of equation \eqref{eq:joint epsilon cdf}. 

\bigskip

For both the $\mathcal{T}$- and $\mathcal{U}$-independence proofs, we use the following two inequalities: First, for all $u \in \supp(U)$,
\begin{align}\label{eq:uniformUpperBound}
	F_{U \mid X}(u \mid x)
		&= \int_{-\infty}^u \frac{\Prob(X=x \mid U=v)}{p_x} \; dF_U(v) \notag\\
		&\leq \int_{-\infty}^u \frac{1}{p_x} \; dF_U(v) \notag\\
		&= \frac{F_U(u)}{p_x}.
\end{align}
The first line follows by lemma \ref{lemma:cdfAsIntegralOfPropensityScore}. The second line follows by $\Prob(X=x \mid U=v) \leq 1$. Second, for all $u \in \supp(U)$,
\begin{align}\label{eq:uniformLowerBound}
	F_{U \mid X}(u \mid x)
		&= \int_{-\infty}^u \frac{\Prob(X=x \mid U=v)}{p_x} \; dF_U(v) \notag \\
		&= 1 - \int_u^\infty \frac{\Prob(X=x \mid U=v)}{p_x} \; dF_U(v) \notag \\
		&\geq 1 - \int_u^\infty \frac{1}{p_x} \; dF_U(v) \notag \\
		&= 1 + \frac{F_U(u) - 1}{p_x}.
\end{align}
While equations \eqref{eq:uniformUpperBound} and \eqref{eq:uniformLowerBound} both hold for all $u \in \supp(U)$, they are not sharp for all $u$.

\bigskip

\textbf{Part 1.} We show that $F_{U \mid X}(u \mid x) \leq \overline{F}_{U \mid X}^{\mathcal{T}}(u \mid x)$ for all $u \in \supp(U)$. If $u \leq Q_U(p_xF_U(a))$, the upper bound holds by equation \eqref{eq:uniformUpperBound}.
Second, if $u \in [Q_U(p_xF_U(a)),a]$, then $F_{U \mid X}(u \mid x) \leq F_{U \mid X}(a \mid x) = F_U(a)$ since $F_{U \mid X}(\cdot \mid x)$ is nondecreasing and by $\mathcal{T}$-independence. Third, if $u\in [a,b]$, then $F_{U \mid X}(u \mid x) = F_U(u)$ by $\mathcal{T}$-independence. Fourth, if $u \in [b, Q_U(p_x + F_U(b)(1-p_x))]$, then 
\begin{align*}
	F_{U \mid X}(u \mid x)
		&= \int_{-\infty}^u \frac{\Prob(X=x \mid U=v)}{p_x} \; dF_U(v)\\
		&= \int_{-\infty}^b \frac{\Prob(X=x \mid U=v)}{p_x} \; dF_U(v)
			+ \int_b^u \frac{\Prob(X=x \mid U=v)}{p_x} \; dF_U(v)\\
		&\leq F_{U \mid X}(b \mid x) + \int_b^u \frac{1}{p_x} \; dF_U(v)\\
		&= F_U(b) + \frac{F_U(u) - F_U(b)}{p_x}.
\end{align*}

Finally, for all $u \in \supp(U)$, $F_{U \mid X}(u \mid x) \leq 1$. In particular, this holds for $u \geq Q_U(p_x + F_U(b)(1-p_x))$.

\bigskip

\textbf{Part 2.} We show that $F_{U \mid X}(u \mid x) \geq \underline{F}_{U \mid X}^{\mathcal{T}}(u \mid x)$ for all $u \in \supp(U)$. First, $F_{U \mid X}(u \mid x) \geq 0$ for all $u \in \supp(U)$. In particular, this holds for $u \leq Q_{U}((1-p_x)F_U(a))$. Second, if $u \in [Q_{U}((1-p_x)F_U(a)),a]$, then 
\begin{align*}
	F_{U \mid X}(u \mid x)
		&= \int_{-\infty}^u \frac{\Prob(X=x \mid U=v)}{p_x} \; dF_U(v)\\
		&= \int_{-\infty}^a \frac{\Prob(X=x \mid U=v)}{p_x} \; dF_U(v)
			- \int_u^a \frac{\Prob(X=x \mid U=v)}{p_x} \; dF_U(v)\\
		&\geq F_{U \mid X}(a \mid x) - \int_u^a \frac{1}{p_x} \; dF_U(v)\\
		&= F_U(a) + \frac{F_U(u) - F_U(a)}{p_x}.
\end{align*}
Third, if $u\in [a,b]$ then $F_{U \mid X}(u \mid x) = F_U(u)$ by $\mathcal{T}$-independence. Fourth, if $u\in[b,Q_U(p_xF_U(b) + (1-p_x))]$, then $F_{U \mid X}(u \mid x) \geq F_{U \mid X}(b \mid x) = F_{U}(b)$. Finally, if $u \geq Q_U(p_xF_U(b) + (1-p_x))$, the lower bound holds by equation \eqref{eq:uniformLowerBound}.

\bigskip

\textbf{Part 3}. To prove sharpness, we must construct a joint distribution of $(U,X)$ consistent with assumptions 1--4 and which yields the upper bound $\overline{F}_{U \mid X}^\mathcal{T}(\cdot \mid x)$. And likewise for the lower bound $\underline{F}_{U \mid X}^\mathcal{T}(\cdot \mid x)$. This yields equation \eqref{eq:joint epsilon cdf} for $\varepsilon = 0$ and $\varepsilon = 1$. By taking convex combinations of these two joint distributions we obtain the case for $\varepsilon \in (0,1)$.

The marginal distributions of $U$ and $X$ are prespecified. Hence to construct the joint distribution of $(U,X)$ it suffices to define conditional distributions of $U \mid X$. We define these conditional distributions by the bound functions themselves, $\underline{F}_{U \mid X}^{\mathcal{T}}(u \mid x)$ and $\overline{F}_{U \mid X}^{\mathcal{T}}(u \mid x)$. These functions are non-decreasing, right-continuous, and have range $[0,1]$. Hence they are valid cdfs. They also satisfy $\mathcal{T}$-independence. 
These properties are preserved by taking convex combinations, and hence $\varepsilon \underline{F}_{U \mid X}^{\mathcal{T}}(u \mid x) + (1-\varepsilon)\overline{F}_{U \mid X}^{\mathcal{T}}(u \mid x)$ is also a valid cdf that satisfies $\mathcal{T}$-independence for any $\varepsilon \in [0,1]$ and $x\in\{0,1\}$. 
Finally, we show that these cdfs are consistent with the marginal distribution of $U$, and can satisfy both components of equation \eqref{eq:joint epsilon cdf} simultaneously. To see this, we compute
\begin{align*}
	&p_x \overline{F}_{U \mid X}^\mathcal{T}(u \mid x) + (1-p_x) \underline{F}_{U \mid X}^\mathcal{T}(u \mid 1-x) \\
	&=
	\begin{cases}
	p_x\dfrac{F_U(u)}{p_x}
		&\text{if $u\leq Q_{U}(p_xF_U(a))$}\\
	p_xF_U(a) + F_U(u) - F_U(a) + F_U(a)(1-p_x)
		&\text{if $Q_{U}(p_xF_U(a)) \leq u \leq a$}\\
	p_xF_U(u) + (1-p_x)F_U(u)
		&\text{if $a\leq u \leq b$} \\
	F_U(u) - F_U(b) + p_xF_U(b) + (1-p_x) F_U(b)
		&\text{if $b\leq u \leq Q_{U}(p_x + F_U(b)(1-p_x))$}\\	p_x + F_U(u) - 1 + (1-p_x)
		&\text{if $Q_{U}(p_x + F_U(b)(1-p_x))\leq u$}
	\end{cases} \\
	&= F_U(u).
\end{align*}
Thus
\begin{align*}
	&p_1 \left[ \varepsilon \underline{F}_{U \mid X}^{\mathcal{T}}(u \mid 1) 
		+ (1-\varepsilon)\overline{F}_{U \mid X}^{\mathcal{T}}(u \mid 1) \right] 
	+ p_0 \left[ (1-\varepsilon)\underline{F}_{U \mid X}^{\mathcal{T}}(u \mid 0)
		+ \varepsilon \overline{F}_{U \mid X}^{\mathcal{T}}(u \mid 0) \right] \\
	&= \varepsilon \left[ p_1 \underline{F}_{U \mid X}^\mathcal{T}(u \mid 1)
		+ p_0 \overline{F}_{U \mid X}^\mathcal{T}(u \mid 0) \right] 
		+ (1-\varepsilon) \left[ p_1 \overline{F}_{U \mid X}^\mathcal{T}(u \mid 1)
		+ p_0 \underline{F}_{U \mid X}^\mathcal{T}(u \mid 0) \right] \\
	&= \varepsilon F_U(u) + (1-\varepsilon) F_U(u) \\
	&= F_U(u).
\end{align*}
\end{proof}

\begin{proof}[Proof of proposition \ref{prop:TauCDFbounds} ($\mathcal{U}$-independence)]
Now we consider the cdf bounds under $\mathcal{U}$-independence, under various cases: 

\bigskip

\textbf{Part 1.} We show $F_{U \mid X}(u \mid x) \leq \overline{F}_{U \mid X}^{\mathcal{U}}(u \mid x)$ for all $u \in \supp(U)$. We do this in two cases.

\medskip

\textbf{Part 1a.} Suppose $(1 - (F_U(b) - F_U(a)))p_x \leq F_U(a)$. First, $F_{U \mid X}(u \mid x) \leq 1$ for all $u \in \supp(U)$. In particular, this holds if $u \geq b$.

Second, if $u \in [a,b]$, then
\begin{align*}
	F_{U \mid X} (u \mid x)
		&= \int_{-\infty}^u \frac{\Prob(X=x \mid U=v)}{p_x} \; dF_U(v) \\
		&= 1- \int_{b}^\infty \frac{\Prob(X=x \mid U=v)}{p_x} \; dF_U(v)
			- \int_{u}^b \frac{\Prob(X=x \mid U=v)}{p_x} \; dF_U(v) \\
		&\leq 1 - \int_{b}^\infty \frac{0}{p_x} \; dF_U(v)
			- \int_{u}^b \frac{p_x}{p_x} \; dF_U(v) \\
		&= 1 - (F_U(b) - F_U(u)).
\end{align*}
Third, if $u \in [Q_U((1-(F_U(b) - F_U(a)))p_x), a]$, then $F_{U \mid X}(u \mid x) \leq F_{U \mid X}(a \mid x) \leq 1 - (F_U(b) - F_U(a))$ where the last inequality follows by our derivation immediately above. Finally, if $u \leq Q_{U}((1-(F_U(b) - F_U(a)))p_x)$, the upper bound holds by equation \eqref{eq:uniformUpperBound}.

\bigskip

\textbf{Part 1b.} Now suppose $(1 - (F_U(b) - F_U(a)))p_x \geq F_U(a)$. First, if $u \leq a$, the upper bound holds by equation \eqref{eq:uniformUpperBound}. Second, if $u \in [a,b]$ then 
\begin{align*}
	F_{U \mid X} (u \mid x)
		&= \int_{-\infty}^u \frac{\Prob(X=x \mid U=v)}{p_x} \; dF_U(v) \\
		&= \int_{-\infty}^a \frac{\Prob(X=x \mid U=v)}{p_x} \; dF_U(v)
			+ \int_{a}^u \frac{\Prob(X=x \mid U=v)}{p_x} \; dF_U(v)\\
		&\leq \int_{-\infty}^a \frac{1}{p_x} \; dF_U(v)
			+ \int_{a}^u \frac{p_x}{p_x} \; dF_U(v)\\
		&= \frac{F_{U}(a)}{p_x} + F_U(u) - F_U(a).
\end{align*}
Third, if $u \in [b,Q_U((F_U(b)-F_U(a))(1-p_x)+ p_x)]$, then 
\begin{align*}
	&F_{U \mid X} (u \mid x) \\
		&= \int_{-\infty}^u \frac{\Prob(X=x \mid U=v)}{p_x} \; dF_U(v) \\
		&= \int_{-\infty}^a \frac{\Prob(X=x \mid U=v)}{p_x} \; dF_U(v)
			+ \int_{a}^b \frac{\Prob(X=x \mid U=v)}{p_x} \; dF_U(v)
			+ \int_{b}^u \frac{\Prob(X=x \mid U=v)}{p_x} \; dF_U(v) \\
		&\leq \int_{-\infty}^a \frac{1}{p_x} \; dF_U(v)
			+ \int_{a}^b \frac{p_x}{p_x} \; dF_U(v)
			+ \int_{b}^u \frac{1}{p_x} \; dF_U(v) \\
		&= \frac{F_U(a) + p_x(F_U(b) - F_U(a)) + F_U(u) - F_U(b)}{p_x}.
\end{align*}
Finally, if $u \geq Q_U((F_U(b)-F_U(a))(1-p_x)+ p_x)$, then $F_{U \mid X}(u \mid x) \leq 1$. 

\bigskip

\textbf{Part 2.} We show that $F_{U \mid X}(u \mid x) \geq \underline{F}_{U \mid X}^{\mathcal{U}}(u \mid x)$ for all $u \in \supp(U)$. We do this in two cases.

\medskip

\textbf{Part 2a.} Suppose $(1 - (F_U(b) - F_U(a)))(1-p_x) \leq F_U(a)$. First, if $u \leq Q_U((1-(F_U(b) - F_U(a)))(1-p_x))$, then $F_{U \mid X}(u \mid x) \geq 0$.
Second, if $u \in [Q_U((1-(F_U(b) - F_U(a)))(1-p_x)),a]$, then
\begin{align*}
	&F_{U \mid X} (u \mid x) \\
		&= \int_{-\infty}^u \frac{\Prob(X=x \mid U=v)}{p_x} \; dF_U(v) \\
		&= 1- \int_{b}^\infty \frac{\Prob(X=x \mid U=v)}{p_x} \; dF_U(v)
			- \int_{a}^b \frac{\Prob(X=x \mid U=v)}{p_x} \; dF_U(v)
			- \int_{u}^a \frac{\Prob(X=x \mid U=v)}{p_x} \; dF_U(v) \\
		&\geq 1- \int_{b}^\infty \frac{1}{p_x} \; dF_U(v)
			- \int_{a}^b \frac{p_x}{p_x} \; dF_U(v)
			- \int_{u}^a \frac{1}{p_x} \; dF_U(v) \\
		&=  \frac{F_U(u) - (1 - (F_U(b) -F_U(a)))(1-p_x) }{p_x}.
\end{align*}

Third, if $u\in [a,b]$, then 
\begin{align*}
	F_{U \mid X} (u \mid x)
		&= \int_{-\infty}^u \frac{\Prob(X=x \mid U=v)}{p_x} \; dF_U(v) \\
		&= 1- \int_{b}^\infty \frac{\Prob(X=x \mid U=v)}{p_x} \; dF_U(v)
			- \int_{u}^b \frac{\Prob(X=x \mid U=v)}{p_x} \; dF_U(v) \\
		&\geq 1 - \int_{b}^\infty \frac{1}{p_x} \; dF_U(v)
			- \int_{u}^b \frac{p_x}{p_x} \; dF_U(v) \\
		&= F_U(u) + \frac{(F_U(b) - 1)(1-p_x)}{p_x}.
\end{align*}
Finally, if $u \geq b$, the lower bound holds by equation \eqref{eq:uniformLowerBound}.

\bigskip

\textbf{Part 2b.} Now suppose $(1 - (F_U(b) - F_U(a)))(1-p_x) \geq F_U(a)$. First, if $u\leq a$ then $F_{U \mid X}(u \mid x) \geq 0$. Second, if $u\in [a,b]$ then 
\begin{align*}
	F_{U \mid X} (u \mid x)
		&= \int_{-\infty}^u \frac{\Prob(X=x \mid U=v)}{p_x} \; dF_U(v) \\
		&= \int_{-\infty}^a \frac{\Prob(X=x \mid U=v)}{p_x} \; dF_U(v)
			+ \int_{a}^u \frac{\Prob(X=x \mid U=v)}{p_x} \; dF_U(v) \\
		&\geq \int_{-\infty}^a \frac{0}{p_x} \; dF_U(v)
			+ \int_{a}^u \frac{p_x}{p_x} \; dF_U(v) \\
		&= F_U(u) - F_U(a).
\end{align*}
Third, if $u\in [b,Q_U(p_x(F_U(b)-F_U(a)) + 1-p_x)]$, then $F_{U \mid X}(u \mid x) \geq F_{U \mid X}(b \mid x) \geq F_{U}(b) - F_U(a)$, where the last inequality follows by our derivation immediately above. Finally, if $u \geq Q_U(p_x(F_U(b)-F_U(a)) + 1-p_x)$ the lower bound holds by equation \eqref{eq:uniformLowerBound}.

\bigskip

\textbf{Part 3.} In this part, we prove sharpness in two steps. First we construct a joint distribution of $(U,X)$ consistent with assumptions 1--4 and which yields the upper bound $\overline{F}_{U \mid X}^\mathcal{U}(\cdot \mid x)$. And likewise for the lower bound $\underline{F}_{U \mid X}^\mathcal{U}(\cdot \mid x)$. This yields equation \eqref{eq:joint epsilon cdf} for $\varepsilon = 0$ and $\varepsilon = 1$. Second we use convex combinations of these two joint distributions to obtain the case for $\varepsilon \in (0,1)$.

The marginal distributions of $U$ and $X$ are prespecified. Hence to construct the joint distribution of $(U,X)$ it suffices to define conditional distributions of $X \mid U$. Specifically, when $(1 - (F_U(b) - F_U(a)))p_x \leq F_U(a)$, define the conditional probability
\[
	\overline{p}_x(u) =
	\begin{cases}
		1 &\text{ for $u < Q_U((1-(F_U(b)-F_U(a)))p_x)$} \\
		0 &\text{ for $u \in [(Q_U((1-(F_U(b)-F_U(a)))p_x),a)$} \\
		p_x &\text{ for $u \in [a,b)$} \\
		0 &\text{ for $u \geq b$.}
	\end{cases}
\]
for $u\in\supp(U)$. This conditional probability is consistent with $\mathcal{U}$-independence. Moreover, by applying lemma \ref{lemma:cdfAsIntegralOfPropensityScore} one can verify that it yields the upper bound $\overline{F}_{U \mid X}^{\mathcal{U}}(\cdot \mid x)$.

When $(1 - (F_U(b) - F_U(a)))p_x \geq F_U(a)$, define
\[
	\overline{p}_x(u) =
	\begin{cases}
		1 &\text{ for $u < a$} \\
		p_x &\text{ for $u \in [a,b)$} \\
		1 &\text{ for $u \in [b,Q_U((F_U(b)-F_U(a))(1-p_x)+ p_x))$} \\
		0 &\text{ for $u \geq Q_U((F_U(b)-F_U(a))(1-p_x)+ p_x)$.}
	\end{cases}
\]
Again, by applying lemma \ref{lemma:cdfAsIntegralOfPropensityScore} one can verify that this conditional probability yields the upper bound $\overline{F}_{U \mid X}^{\mathcal{U}}(\cdot \mid x)$.

Next consider the lower bounds. When $(1 - (F_U(b) - F_U(a)))(1-p_x) \leq F_U(a)$, define
\[
	\underline{p}_x(u) =
	\begin{cases}
		0 &\text{ for } u < Q_U((1-(F_U(b)-F_U(a)))(1-p_x))\\
		1 &\text{ for } u\in[Q_U((1-(F_U(b)-F_U(a)))(1-p_x)),a)\\
		p_x &\text{ for } u\in[a,b)\\
		1 &\text{ for } u\geq b.
	\end{cases}
\]
When $(1 - (F_U(b) - F_U(a)))(1-p_x) \geq F_U(a)$, define
\[
	\underline{p}_x(u) =
	\begin{cases}
		0 &\text{ for } u < a\\
		p_x &\text{ for } u\in[a,b)\\
		0 &\text{ for } u\in[b, Q_U(p_x(F_U(b) - F_U(a)) + 1-p_x))\\
		1 &\text{ for } u\geq Q_U(p_x(F_U(b) - F_U(a)) + 1-p_x).
	\end{cases}
\]
As with the upper bounds, one can verify that these yield the lower bound $\underline{F}_{U \mid X}^{\mathcal{U}}(\cdot \mid x)$. For all of these conditional distributions of $X \mid U$, one can verify that they are consistent with the marginal distribution of $X$:
\[
	\int_{\supp(U)} \overline{p}_x(u) \; dF_U(u) = p_x
	\qquad \text{and} \qquad
	\int_{\supp(U)} \underline{p}_x(u) \; dF_U(u) = p_x.
\]
Thus we have shown that the bound functions are attainable. That is, equation \eqref{eq:joint epsilon cdf} holds with $\varepsilon = 0$ or $1$. Next consider $\varepsilon \in (0,1)$. For this $\varepsilon$, we specify the distribution of $X \mid U$ by the conditional probability $\varepsilon \underline{p}_x(u) + (1-\varepsilon) \overline{p}_x(u)$. This is a valid conditional probability since it is a convex combination of two terms which are between 0 and 1. This conditional probability satisfies $\mathcal{U}$-independence. By linearity of the integral and our results above,
\[
	\int_{\supp(U)} \left[ \varepsilon \underline{p}_x(u) + (1-\varepsilon) \overline{p}_x(u) \right] \; dF_U(u) = p_x
\]
and hence this distribution of $X \mid U$ is consistent with the marginal distribution of $X$. Finally, by lemma \ref{lemma:cdfAsIntegralOfPropensityScore} and linearity of the integral, this conditional probability yields the cdf
\[
	\Prob(U \leq u \mid X=x) = \varepsilon \underline{F}_{U \mid X}^{\mathcal{U}}(u \mid x) + (1-\varepsilon)\overline{F}_{U \mid X}^{\mathcal{U}}(u \mid x),
\]
as needed for each component of equation \eqref{eq:joint epsilon cdf}. To see that each component of equation \eqref{eq:joint epsilon cdf} holds simultaneously, we show that a law of total probability constraint holds. There are two cases to check. First suppose $(1 - (F_U(b) - F_U(a)))(1-p_1)  = (1 - (F_U(b) - F_U(a)))p_0 \leq F_U(a)$. Then
\begin{align*}
	&p_1 \underline{F}_{U \mid X}^\mathcal{U}(u \mid 1) + 	p_0 \overline{F}_{U \mid X}^\mathcal{U}(u \mid 0)\\
 	&=
	\begin{cases}
		0 + F_U(u) \\
		\qquad \text{ for $u < Q_U((1-(F_U(b)-F_U(a)))p_0)$} \\
		\\
		(F_U(u) - (1-(F_U(b) - F_U(a))))p_0 + (1-(F_U(b) - F_U(a)))p_0 \\
		\qquad \text{ for $u \in [Q_U((1-(F_U(b)-F_U(a)))p_0),a)$} \\
		\\
		(F_U(b)-1)p_0 + F_U(u)p_1 + (1 - (F_U(b) - F_U(u))p_0
		&\text{ for $u \in [a,b)$} \\
		F_U(u) - 1 + p_1 + p_0
			&\text{ for $u \geq b$}
	\end{cases}\\
	&= F_U(u).
\end{align*}
Likewise, $p_1 \overline{F}_{U \mid X}^\mathcal{U}(u \mid 1) + 	p_0 \underline{F}_{U \mid X}^\mathcal{U}(u \mid 0) = F_U(u)$. Similar derivations hold for the other case. 
Thus
\begin{align*}
	&p_1 \left[ \varepsilon \underline{F}_{U \mid X}^{\mathcal{U}}(u \mid 1) 
		+ (1-\varepsilon)\overline{F}_{U \mid X}^{\mathcal{U}}(u \mid 1) \right] 
	+ p_0 \left[ (1-\varepsilon)\underline{F}_{U \mid X}^{\mathcal{U}}(u \mid 0)
		+ \varepsilon \overline{F}_{U \mid X}^{\mathcal{U}}(u \mid 0) \right] \\
	&= \varepsilon \left[ p_1 \underline{F}_{U \mid X}^\mathcal{U}(u \mid 1)
		+ p_0 \overline{F}_{U \mid X}^\mathcal{U}(u \mid 0) \right] 
		+ (1-\varepsilon) \left[ p_1 \overline{F}_{U \mid X}^\mathcal{U}(u \mid 1)
		+ p_0 \underline{F}_{U \mid X}^\mathcal{U}(u \mid 0) \right] \\
	&= \varepsilon F_U(u) + (1-\varepsilon) F_U(u) \\
	&= F_U(u).
\end{align*}
\end{proof}

\end{document}